\documentclass[10pt,a4paper,dvipsnames]{article}

\usepackage{fullpage,amssymb,amsmath,amsfonts,amsthm,graphicx,color,float, mathtools,tikz} 
\usepackage{paralist} 
\usepackage[utf8]{inputenc}
\usepackage[english]{babel}
\usepackage[normalem]{ulem}
\usepackage{xcolor}
\usepackage[textsize=scriptsize, backgroundcolor=blue!20!white,bordercolor=red]{todonotes} 
\usepackage{dsfont} 
\usepackage{mathrsfs}
\usepackage{stmaryrd}
\usepackage{cite}

\newtheorem{theorem}{Theorem}[section]
\newtheorem{lemma}[theorem]{Lemma}

\newtheorem{corollary}[theorem]{Corollary}

\theoremstyle{definition}
\newtheorem{definition}[theorem]{Definition}
\newtheorem*{remark*}{Remark}

\newcommand{\E}{\mathbb E}
\newcommand{\R}{\mathbb R}

\newcommand{\N}{\mathbb N}
\newcommand{\dd}{\, \mathrm{d}}
\newcommand{\norm}[1]{\lVert#1\rVert}

\numberwithin{equation}{section}
\allowdisplaybreaks 

\title{Microscopic derivation of Vlasov equation with compactly supported pair potentials }
\author{Manuela Feistl, Peter Pickl}

\begin{document}
\maketitle

\begin{abstract}
We present a probabilistic proof of the mean-field limit and propagation of chaos of a N-particle system in three dimensions with compactly supported pair potentials of the form $N^{3\beta-1} \phi(N^{\beta}x)$ for $\beta\in\left[0,\frac{1}{7}\right)$ and $\phi\in L^{\infty}(\mathbb{R}^3)\cap L^1(\mathbb{R}^3)$. In particular, for typical initial data, we show convergence of the Newtonian trajectories to the characteristics of the Vlasov-Dirac-Benney system with delta-like interactions. The proof is based on a Gronwall estimate for the maximal distance between the exact microscopic dynamics and the approximate mean-field dynamics. Thus our result leads to a derivation of the Vlasov-Dirac-Benney equation from the microscopic $N$-particle dynamics with a strong short range force.
\end{abstract}

\section{Introduction}
Many physical phenomena have to be explained by purely kinetic mechanisms. Recently, the semicollisional regimes have gained a great deal of interest. 
Solving the Newtonian equations of motion for numerous interacting particles becomes infeasible as the particle count increases.
To circumvent this problem but still obtain some kind of solution one can change the level of description by derivating an effective equation, which describes a continuous mass density which effectively still represents the same situation from a macroscopic point of view. Under certain assumptions on the initial density $k_0$ the characteristics of the effective equations provide typically a very good approximation of the $N$-particle trajectories in the limit as the particle number tends to infinity if their initial positions are independent and identically distributed with respect to the density $k_0$. 

\subsection{Previous results}
Numerous models have been proposed in the literature which reproduce the kinetic effects in effective equations describing gases or fluids. One example of such an effective equation goes back to Vlasov \cite{VLASOV}, which has been derived with mathematical rigour by Neunzert and Wick in 1974 \cite{NeunzertWick}. Classical results of this kind are valid for Lipschitz-continuous forces \cite{BraunHepp,Dobrushin}.
One difficulty is handling clustering of particles for singular interactions (see \cite{SpohnBook}) like Coulomb or Newtons gravitational force.
Hauray and Jabin could include singular interaction forces scaling like $1/\vert q \vert^{\lambda}$ in three dimensions with ${\lambda} < 1$ \cite{Hauray2007} and later the physically more interesting case with ${\lambda}$ smaller but close to $2$ with a lower bound on the cut-off at $q = N^{-1/6}$ \cite{Hauray2013} but still in the deterministic setting. Therefore they had to choose quite specific initial conditions, according to the respective $N$-particle law.
The last deterministic result we like to mention  in the Coulomb interaction setting is \cite{Kiessling}, which is valid for repulsive pair-interactions and assumes no cut-off, but instead a bound on the maximal forces of the microscopic system.
One major difference to our work is that the results rely on deterministic initial conditions, even if some of them are formulated probabilistically.
In contrast to the previous papers Boers and Pickl \cite{Peter} derive the Vlasov equations for stochastic initial conditions with interaction forces scaling like $|x|^{-3\lambda+1}$ with $(5/6<\lambda<1)$. They were able to obtain a cut-off  as small as $ N^{-\frac{1}{3}}$, which is the typical inter particle distance. Furthermore, Lazarovici and Pickl \cite{Dustin} extended the method in \cite{Peter} exploiting the second order nature of the dynamics and introducing an anisotropic scaling of the relevant metric to include the Coulomb singularity and they obtained a microscopic derivation of the Vlasov-Poisson equations with a cut-off of $N^{-\delta}$  $(0<\delta<\frac{1}{3})$.
More recently, the cut-off parameter was reduced to as small as $N^{-\frac{7}{18}}$ in \cite{grass} by using the second order nature of the dynamics and a having a closer look at the collisions which could occur.
There are some result from Ölschläger for potentials of the form $N^{3\beta-1} \phi(N^{\beta})$ in the classical setting like \cite{Oelschläger Brown}, including Brownian motion, or \cite{Oelschläger Hydro} which addresses the derivation of the continuity equation in the monokinetic setting.
To our knowledge there is no derivation in probability of Vlasov-Dirac-Benney equation for compactly supported pair potentials in the literature. 

\subsection{Present result}
The strategy presented in this paper uses stochastic initial conditions, as it is based on the technique proposed in \cite{Peter,Dustin,grass} but in contrast to these we use a highly singular force term instead of Coulomb interaction.
The usual methods for deriving effective descriptions for microscopic dynamics fail here because there exists no kind of Lipschitz condition for the force. The whole interaction has a Lipschitz constant that depends on $N$. In all mentioned papers above in most of the cases the potential is nice and smooth. The main part of these proofs considers cases were the force is Lipschitz-continuous and the cases where the force is not is negligible due to probabilistic arguments. In our case we have an interaction that can be felt by the leading order and has additionally a large derivative. So the force in this paper is not Lipschitz-continuous at all.
With regard to our goal to derive the Vlasov-Dirac-Benney equation from the microscopic Newtonian $N$-particle dynamics we compare the $N$-particle Newtonian flow with the effective flow given by the macroscopic equation in the limit $N\rightarrow\infty$ for a pair potentials of the form 
\begin{align*}
\phi_N^{\beta}=N^{3\beta-1} \phi(N^{\beta})
\end{align*}
with $\beta\in\left[0,\frac{1}{7}\right)$ and $\phi \in L^{\infty}(\mathbb{R}^3)\cap L^1(\mathbb{R}^3)$. Our system is between collision and mean-field behaviour and describes physical situations in which the interaction has a long range since the typical distance between particles is of order $N^{\frac{1}{3}}$. On the one hand, the interaction force is collision-like, so that interactions only rarely take place, i.e. one particle interacts only with a selection of particles of order $\gg 1$ but $\ll N$ and not with all particles. On the other hand, the behaviour can be described with the mean-field approach. The system couples strongly but is localized so an mean-field approach can be applied.
The effect of the instabilities will be much more drastic.
In the following we will prove that the measure of the set where the maximal distance of the Newtonian trajectory and the mean-field trajectory is large gets vanishingly small as $N$ increases.
\subsection{The microscopic model - Newtonian motion of particles and the Newtonian flow}
We consider a classical $N$-particle system subject to Newtonian dynamics interacting through a pair interaction force. 
Our system is distributed as a trajectory in phase space $\mathbb{R}^{6N}.$ We use the notation $X=(Q,P)=(q_1,\hdots,q_N,p_1,\hdots,p_N)$, where $(Q)_j= q_j\in \mathbb{R}^{3}$ denotes the one-particle position and $(P)_j=p_j \in \mathbb{R}^{3}$ stands for its momentum.
The Hamiltonian, the operator corresponding to the total energy of the system, is given by
\begin{align*}
 H_N(X)=\sum_{j=1}^N \frac{p_j^2}{2m}+\sum_{1\leq j<k\leq N} \phi^\beta(q_j-q_k),
\end{align*}
with $\phi \in L^{\infty}(\mathbb{R}^3)\cap L^1(\mathbb{R}^3)$ and $x_j=(q_j,p_j)\in\mathbb{R}^6.$ 
As long as the conditions on the solution of the effective equation are valid an external potential can be added to the Hamiltonian, but it does not affect the derivation of the equation as it has the same impact on all particles regardless of their distribution.
Since we will consider differences between exact Newtonian dynamics and mean-field dynamics we will omit the external potential without loss of generality.
Setting the mass $m=1$ leads us to the equations of motion, which determines the particle trajectories
\begin{align*}\label{eff.flow}
\begin{cases}
\dot{q_j}=\frac{\partial H}{\partial p_j}=p_j\\
\dot{p_j}=\frac{\partial H}{\partial q_j}= -\sum_{j=1}^{N} \nabla_{q_i} \phi_N^{\beta}(q_j-q_k)=-\sum_{j=1}^{N}\frac{1}{N}f_{N}^{\beta}(q_j-q_k).\\
\end{cases}
\end{align*}
The interaction potential $\phi_N^{\beta}$  in dimension three is defined by
\begin{align*}
\phi_N^{\beta}=N^{3\beta-1}\phi(N^{\beta}x), \hspace{1cm} \beta \in \lbrack 0, \frac{1}{7}),
\end{align*}
for some bounded spherical symmetric $\phi:\mathbb{R}^{3}\rightarrow \mathbb{R}$ with $\nabla \phi(0)=0$ and $f_N^{\beta}$ denotes the pair interaction force for the system. 
The factor $\beta\in\mathbb{R}$ determines the scaling behaviour of the interaction and depending on how one chooses $\beta$ one get another hydrodynamic equation. Usually  $\phi_N^\beta$
scales with the particle number such that the total
interaction energy scales in the same way as the total kinetic energy of the $N$ particles, so that the $L^1$-norm of  $\phi_N^\beta$ is proportional to $N^{-1}$ with
$\phi_N^\beta(x)= N^{-1+3\beta} \phi(N^\beta x)$ for $\phi\in L^{\infty}(\mathbb{R}^3)\cap L^1(\mathbb{R}^3)$. 
Note that we assume that $\phi$ and thus $f$ will be independent of the momentum.

The case $\beta=0$ was also studied in \cite{BraunHepp}. The strength of the interaction is of order $\frac{1}{N}$ and hence the equations of motion consider a weakly interaction gas.
For positive $\beta$ the support of the potential shrinks and the strength of the interaction increases with $N$ and $\lim_{N\rightarrow\infty}\phi_N=\delta_0$.
Thus the case $\beta=\frac{1}{3}$ describes in contrast a strong interaction process. The interaction strength is of order $1$ but two particles only interact when their distance is of order of the typical inter particle distance in $\mathbb{R}^3$.
As long as $\beta<1/3$ the mean-field approximation is from the heuristical point of view not surprising because the interaction potentials overlap as the typical particle distance has approximately the size of $N^{\frac{1}{3}}$.
This mean inter-particle distance is consequently smaller than the range of the interaction.
Hence, on average, every particle interacts with many other particles, and the interactions are weak since $N^{-1}N^{3\beta}\rightarrow 0$ as $N\rightarrow \infty$.
As long as the correlations are sufficiently mild, the law of large numbers gives that the interaction can be replaced by its expectation value, the so-called mean-field.

The potential gradient $\nabla_{q}\phi_N^{\beta}=f_N^{\beta}$ determines the pair interaction function, which has the form
$f_N^{\beta}=N^{4\beta}l(N^{\beta}q)$, for some $l\in L^{\infty}(\mathbb{R}^3)\cap L^1(\mathbb{R}^3)$.
\begin{definition}\label{force f}
For $N\in \mathbb{N}\cup\lbrace \infty \rbrace$ and a smooth function $l\in  W^{1,\infty}(\mathbb{R}^3)\cap W^{1,1}(\mathbb{R}^3)$, vanishing at infinity with bounded derivatives, the interaction force $f_N^{\beta}:\mathbb{R}^3\rightarrow \mathbb{R}^3$ is given by 
\begin{align*}
f_N^{\beta}(q)= N^{4\beta}l(N^{\beta}q)
\end{align*}
with $0<\beta < \frac{1}{7}$ .
\end{definition}
Analogously the total force of the system is given by $F:\mathbb{R}^{6N}\rightarrow\mathbb{R}^{3N},$ were the force exhibited on a single coordinate $j$ is given by $(F(X))_j:=\sum_{i\neq j}\frac{1}{N}f_N^{\beta}(q_i-q_j)$.

By introducing the N-particle force we can characterize the Newtonian flow as a solution of the next equation.
As the vector field is Lipschitz for fixed $N$ we have global existence and uniqueness of solutions and hence a N-particle flow.

\begin{definition}\label{Def:Newtonflow}
The Newtonian flow $\Psi_{t,s}^{N}(X)=(\Psi_{t,s}^{1,N}(X)
,\Psi_{t,s}^{2,N}(X))$ on $\mathbb{R}^{6N}$ is defined by the solution of:
\begin{align}
\frac{\dd}{\dd t} \Psi_{t,s}^{N}(X)=V(\Psi_{t,s}^{N}(X)) \in \mathbb{R}^{3N}\times\mathbb{R}^{3N}
\end{align}
where $V$ is given by $V(X)=(P,F(X))$.
\end{definition}
The crucial observation is that the force looks like the empirical mean of the continuous function $\nabla \phi (q_j-\cdot)$ of the random variable $q_j.$ In the limit $N\rightarrow\infty,$ one might expect this to be equal to the expectation value of $\nabla \phi$ given by the convolution  $f_N^{\beta} * \tilde{k}_t$ where $\tilde{k}_t(q,p)$ denotes the mass density at $q$ with momentum $p$ at time $t$.
In the following section we explain the general strategy that we follow in this paper which is based on \cite{Peter,Dustin,grass}.
\subsection{Heuristics and scratch of the proof}\label{Heuristics}
In order to translate the microscopic system to the macroscopic system or vice versa there exist two common techniques.
For $^{N}\Psi_{t,0}(X) = (q_i(t),p_i(t))_{i=1,..,N}$, one can define the corresponding microscopic or empirical density by
\begin{equation*} \mu^N_t[X] = \mu^N_0[\Psi_{t,0}(X)] := \frac{1}{N} \sum \limits_{i=1}^N \delta(\cdot - q_i(t)) \delta(\cdot - p_i(t)).
\end{equation*}
By changing the level of description one can consider an equation for a continuous mass density which describes the same situation but from a macroscopic point of view.  
Furthermore the microscopic force can be written as
\begin{equation*} \frac{1}{N}\sum\limits_{j =1}^N f_N(q_i - q_j) = f_N * \mu^N_t[X] (q_i). \end{equation*}
This relation is often used to translate the microscopic dynamics into a Vlasov equation, allowing to treat $\mu^N_t[X]$ and $k_t$ on the same footing to proof the validity of the common physical descriptions.
For typical $X$, the empirical density $\mu^N_t[X]$ and the solution $k_t$ of the Vlasov equation, which describes the time evolution of the distribution function of plasma consisting of charged particles, are close to each other as $N \to \infty$.
The underlying technique of this paper does it the other way around.
We translate the density $k_t$ into a trajectory.
So one can say that the central idea of our strategy is to sample the regularised mean-field dynamics along trajectories with random initial conditions an then control the difference between mean-field trajectories and the true microscopic trajectories in terms of expectation.
So we will not translate the trajectory into an density how it is often done.
Instead of summing up all iteration terms one expects a single particle to feel only the mean-field produced by all particles together. But as there is only the external force $f$, the time evolution of $k$ is dictated by continuity equation on $\mathbb{R}^6$ and by inserting the expectation value form above it leads us to the partial non-linear Vlasov type equation, which solution theory is also studied for delta like interaction forces \cite{HanKwanRousset, Nouri, Bardos-Besse}.
\subsubsection{Construction of the mean-field force}
To construct the mean-field force as mentions above we use the following strategy which can be heuristically explained as follows. We will split the universe in $j $ boxes of the same volume so that $n_j$ particles are in each box.
As the density is the number of particles per volume we get $k_t(q,p)=\frac{n_j}{V_jN}$ and the force acting on one particle can be written as
\begin{align*}
\bar{f}(q)=\sum_j \frac{n_j}{N}f(q-q_j)=\sum_j V_j k_t(q_j,p_j)f(q-q_j).
\end{align*}
One can read this as a Riemann sum and so it can be written as
\begin{align*}
\approx \int k_t(q,p)f(q-q_j)d^3pd^3q_j=k\ast f(q)
\end{align*}
which is the convolution of $k$ and $f$ in the $q$-coordinate.
The mean-field particles move independently, because we use the same force for every particle and we do not have pair interactions, which lead to correlations. Thus each particle has its own force-term.
In summary, for fixed $k_0$ and $N \in \mathbb{N}$, we consider for any initial configuration $X \in \mathbb{R}^{6N}$  two different time-evolutions: $\Psi_{t,0}^{N}(X)$, given by the microscopic equations and $\Phi_{t,0}^{N}(X)$, given by the time-dependent mean-field force generated by $f^N_t$. We are going to show that for typical $X$, the two time-evolutions are close in an appropriate sense. 
In other words, we have non-linear time-evolution in which $\varphi^N_{t,s}(\cdot\,; f_0)$ is the one-particle flow induced by the mean-field dynamics with initial distribution $k_0$, while, in turn, $k_0$ is transported with the flow $\varphi^N_{t,s}$. Due to the semi-group property $\varphi^N_{t,s'}\circ\varphi^N_{s',s} = \varphi^N_{t,s}$ it generally suffices to consider the initial time $s=0$.
 
\subsubsection{Quantifying the accuracy of the mean-field description}

We want to show that the time derivative of the distance $d_t d(\Psi_t, \Phi_t)$ fulfils a Gronwall inequality.
If $|f|_L < \infty$ it is easy to check, but most physically interesting cases are not Lipschitz continuous. 
For technical reasons it is useful to distinguish two cases $\|\Psi_t-\Phi_t\|\leq N^{-\gamma}$ and $\|\Psi_t-\Phi_t\|> N^{-\gamma}$ for $\gamma>0$.
So we introduce a stochastic process of the following form
\begin{align}\label{J_t heuristic}
J_t := \min  \{ 1, N^{\gamma} \|\Psi_t-\Phi_t \|_{\infty} \}.
\end{align}
The stochastic process in the real prove is slightly different from the one shown above but the idea of the proof can also be understood in the simplified form.
The process $J_t$ helps us to establish a Gronwall type argument of the following kind. For all $t \in \mathbb{R}^{+}$ the expectation $\mathbb{E}(J_t)$ value of $J_t$ tends to zero if $ \mathbb{E}(J_0)$ tends to zero. More precisely we will estimate
\begin{align*}
d_t \mathbb{E}(J_t)\le C(\mathbb{E}(J_t)+\sigma_N(1))
\end{align*}
to receive 
\begin{align*}
\mathbb{E}(J_t)\le e^{Ct}(\mathbb{E}(J_0)+\sigma_N(1)).
\end{align*}

This is useful to model the underlying problem because, if $\mathbb{E}(J_t)$ is small, then the probability to hit $1$ is small, that means that the probability $\mathbb{P}(A)$ for $A=\lbrace|\Psi_{s,0}^{N}(x)-\Phi_{s,0}^{N}(X)|\geq N^{-\gamma}\rbrace$ is small, too.
If  $ \mathbb{E}(J_0)\rightarrow 0$ and for all $t \in \mathbb{R}^{+}$ it holds that $d_t \mathbb{E}(J_t)\le C(\mathbb{E}(J_t)+\sigma_N(1))$
we can show with Gronwalls Lemma, that $\mathbb{E}(J_t)\le e^{Ct}(\mathbb{E}(J_0)+\sigma_N(1))$ and so we get $\mathbb{E}(J_t) \rightarrow 0$. Notice that we choose the same initial conditions for $\Psi$ and for $\Phi$, so $J_0=0$.

It is much better to do a Gronwall estimate on $J_t$ than directly on $P(A)$ because you need some kind of smoothness for the derivative.
Each probability of the set $A$  can be translated into an expectation value of the characteristic function with $\mathbb{E}(\chi_{A}),$ but the stochastic process $J_t$ starts to decline at the boundary of $A$ smoothly.
Both descriptions are basically the same, apart from the superiority of $J_t$ in the later proof.
The cut-off in the Definition of $J_t$ has been chosen at $1$, because if $J_t$ is smaller than $1$ it is directly implied, that $|\Psi_{t,0}^{N}(X)-\Phi_{t,0}^{N}(X)|_{\infty}< N^{-\gamma}$ by the definition of $J_t.$

In order to estimate the time derivative of $E(J_t)$ we see that the inequality $\frac{\dd}{\dd t}\mathbb{E}(J_t)\leq C(\mathbb{E}(J_t)+o_N(1))$ is trivial because the random variable $J_t$ has reached its maximum, the value $1$.
The configurations where $J_t$ is maximal, that is $|\Psi_{t,0}^{N}(X)-\Phi_{t,0}^{N}(X)|_{\infty}\geq N^{-\gamma}$ are irrelevant for finding an upper bound of $\mathbb{E}_{0}(J_{t+dt})-\mathbb{E}_0(J_t)$. 
The set of such configurations will be called $\mathcal{A}_{t}$ and we can see that the expectation value $\mathbb{E}(J_{t+dt}-J_t)$ restricted on the set $\mathcal{A}_{t}$ is less or equal $0$.
\subsubsection{Sketch of the proof}
As mentioned above we use the deviation into configurations belonging to $\mathcal{A}_{t}$ or respectively to $\mathcal{A}^c_{t}$ to estimate the expectation value
\begin{align*}
d_t J_t= \lim_{dt\rightarrow 0} \frac{\mathbb{E}(J_{t+dt}-J_t)}{dt}\le \mathbb{E}(|\frac{\dd}{\dd t}J_t|)=\mathbb{E}(|\dot{J}_t|\mathcal{A}_{t})+\mathbb{E}(|\dot{J}_t|\mathcal{A}_{t}^{C}).
\end{align*}
If $X \in \mathcal{A}_{t}$ we have $\|\Psi_t-\Phi_t\| > N^{-\gamma}$ and by the definition of the help process we have $J_t (X)=1$ and consequently $J_{t+dt}(X)\le 1$. This provides $\mathbb{E}(|\dot{J}_t|\mathcal{A}_{t})=0$.
Furthermore for $\mathbb{E}(|\dot{J}_t||\mathcal{A}_{t}^c)$ one can estimate
\begin{align}
\begin{split}
\mathbb{E}(|\dot{J}_t||A^c)&=\mathbb{E}(\|\dot{\Psi_t-\Phi_t}\|_{\infty}|A^c)N^{\gamma}\\
&\le \mathbb{E}(\|F(\Psi_t)-\bar{F}(\Phi_t)\|_{\infty}|A^c)N^{\gamma}+\mathbb{E}(J_t|A^c)\\
&\le \mathbb{E}(\|F(\Psi_t)-F(\Phi_t)\|_{\infty})N^{\gamma}+\mathbb{E}(\|F(\Phi_t)-\bar{F}(\Phi_t)\|_{\infty}|A^c)N^{\gamma}+\mathbb{E}(J_t|A^c)
\end{split}
\end{align}
The last addend is trivially bounded by $\|\dot{\Psi_t-\Phi_t}\|_{\infty}$ due to Newtons law.
To estimate the other addends we will introduce a version of law of large numbers and use the Markov inequality.
Therefore the first addend needs some preparatory work because we can not apply law of large numbers directly. For this reason we will estimate the difference by a mean value argument
\begin{align*}
\|F(\Psi)-F(\Phi)\|_{\infty}N^{\gamma}=\|\frac{1}{N}\sum_{j\neq k}f(q_{j}^{\Psi}-q_{k}^{\Psi})-f(q_{j}^{\Phi}-q_{k}^{\Phi})\|_{\infty}N^{\gamma}\leq\sum_{j\neq 1}g(q_{1}^{\Phi}-q_{j}^{\Phi})\cdot 2\underbrace{\|\Psi-\Phi\|_{\infty}}_{\leq N^{-\gamma}} .
\end{align*}
The last term is independent from $\Psi$ and because of that stochastically independent.
 The estimation of $g$ provided by the law of large numbers determines the choice of the parameter $\beta$ due to the occurring variance term. 
The variance is given by the integral over $g^2$ which is of order of $N^{10\beta-3\beta}$ due to integration by substitution.

\subsection{The macroscopic model: The Vlasov-Dirac-Benney equation and the characteristic flow of the mean-field system}
Looking for a macroscopic law of motion for the particle density leads us to a continuity equation more specifically to the  Vlasov-Dirac-Benney equation owing to the delta like potential.
The theory of solution for each dimension is given for smooth data with finite Sobolev regularity. To state a characterization of solutions we shall  first introduce  the  Penrose stability condition \cite{Penrose}  for homogeneous equilibria $k(p)$. 
\begin{definition}
For $p\mapsto k(q,p)$ the Penrose function is defined by
$$ \mathcal{P}(\gamma, \tau, \eta, k)= 1 -  \int_{0}^{+ \infty} e^{-(\gamma + i \tau)s}\, \frac{i \eta}{1 + |\eta|^2}\cdot  ( \mathcal{F}_{p} \nabla_{p}  k)(\eta s ) \, d s, \quad \gamma >0,\, \tau \in \mathbb{R}, \,\eta \in \mathbb{R}^d\backslash\{0\}$$
where $ \mathcal{F}_{p}$ denotes the Fourier transform in momentum coordinate $p$.
The profile $k$ satisfies the $c_{0}$  Penrose stability condition if 
\begin{equation}
\label{Penrose}
\inf_{(\gamma, \tau, \eta) \in (0,+\infty) \times \R \times \mathbb{R}^d\setminus\{0\}} \left| \mathcal{P}(\gamma, \tau, \eta, k) \right| \geq c_{0}.
\end{equation}
\end{definition} 
These assumptions are for example satisfied in a small data regime, for “one bump” profiles in $d=1$ and also for any radial non-increasing functions in any dimension.
Furthermore the solution theory requires for the initial density $k_0(q,p) \in \mathcal{H}^{2m}_{2r}$.
Therefore we introduce the weighted Sobolev norms for  $k \in \N, r \in \R$ given by
\begin{align*}
\| k \|_{\mathcal{H}^{k}_{r}} := \left(\sum_{|\alpha| + |\beta| \leq k} \int_{\mathbb{T}^d} \int_{\R^d} (1+ |p|^2)^{r} |\partial^\alpha_x \partial^\beta_p k|^2 \, dp dq \right)^{1/2},
\end{align*}
where $\alpha = (\alpha, \cdots, \alpha_d), \beta = (\beta_1, \cdots, \beta_d) \in \N^d$,
$
|\alpha|= \sum_{i=1}^d \alpha, |\beta|= \sum_{i=1}^d \beta_i,
$ and
$
\partial^\alpha := \partial^{\alpha_1}_{q_1} \cdots \partial^{\alpha_d}_{q_d},\partial^\beta_p := \partial^{\beta_1}_{p_1} \cdots \partial^{\beta_d}_{p_d}.$

Let the initial density $k_0(q,p) \in \mathcal{H}^{2m}_{2r}$ with $2m>4+\frac{d}{2}+\lfloor\frac{d}{2}\rfloor$,  $2r>\max(d,2+\frac{d}{2})$ be such that $p\mapsto k_0(q,p)$ satisfies the $\frac{c_{0}}{2}$ Penrose stability condition we consider the corresponding mean-field equation, namely the Vlasov-Dirac-Benney equation
\begin{equation} \label{Vlasov}
\left\{
\begin{aligned}
& \partial_t k \: + \: p \cdot \nabla_q k \: + \:   E \cdot \nabla_p k \:  = \: 0, \\
& E = - \nabla_q \phi, \quad  \phi=\int_{\R^d} k \, dp, \\
&k \vert_{t=0} = k_0(q,p).
\end{aligned}
\right. 
\end{equation}

This equation describes a plasma of identical charged particles with electrostatic or gravitational interactions.
For a fixed initial distribution $k_0 \geq 0$ and $\int k_t(q,p)d^3p=\tilde{k}_t(q)$ we denote by $k^N_t$, the probability density of a particle which, at
time t, occupies the position $q\in\R^3$ and the velocity $q\in\R^3$. It states a solution of  \eqref{Vlasov} with initial datum $k^N_t(0, \cdot,\cdot) = k_0$. 

But due to its highly singular nature, the solution theory is not trivial. One existence result we want to mention here is \cite{Besse}. The authors consider so called water bags, which are piecewise constant functions, as initial data. Another result is \cite{Nouri} which proves the existence for short times of analytical solutions in dimension one.
Bardos and Besse \cite{Bardos-Besse} could show, in dimension one, that the problem  is wellposed for functions that for all $q$ have the shape of one bump.
Later Han-Kwan and Rousset \cite{HanKwanRousset} showed the well-posedness of the system \eqref{Vlasov} in any dimension for smooth data with finite Sobolev regularity  such that for every $q$, $k$ satisfies a Penrose stability condition. More precisely they proved that there exists $T>0$ for which there is a unique solution to (\ref{Vlasov}) with initial condition $k^0$ such that   $k \in \mathcal{C}([0, T], \mathcal{H}^{2m-1}_{2r})$,  $ \phi \in L^2([0, T], H^{2m})$ and  satisfies  the $c_{0}/2$ Penrose condition  for every $t\in [0, T]$.

In this paper we will assume existence  of a $C^{\infty}$ solution of the Vlasov-Dirac-Benney equation and derive it from the microscopic $N$-particle dynamics.
The characteristics of Vlasov equation similar to Definition \ref{Def:Newtonflow} are given by the following system of Newtonian differential equations
\begin{align}
 \begin{cases}
&\frac{dq}{dt}=p \\
&\frac{dp}{dt}=f*\widetilde{k}_t(q) \label{Vlasov.flow}
\end{cases}
\end{align}
where $\widetilde{k}_t$ denotes the previously introduced `spatial density'.
We now introduce the effective one-particle flow $(\varphi_{t,s}^{N})_{t\leq s}$ for any probability density $k_0:\mathbb{R}^{6}\rightarrow \mathbb{R}_{0}^{+}$ and lift it up to the N-particle phase space.

\begin{definition}
Let $k_0:\mathbb{R}^{6}\rightarrow \mathbb{R}_{0}^{+}$ be a probability density and $k:\mathbb{R}\times \mathbb{R}^{6}\rightarrow \mathbb{R}_{0}^{+}$, which gives for each time $t$ the effective distribution function time-evolved with respect to $\varphi_{t,s}^{N}: k(0,\cdot)=k_0$ and
\begin{align*}
k_{t}^{N}(x):=k^{N}(t,x)=k_0(\varphi_{0,t}^{N}(x)).
\end{align*}
For $x=(q,p)$, the effective flow $\varphi_{t,s}^{N}$ itself is defined by
\begin{align*}
\frac{\dd}{\dd t}\varphi_{t,s}^{N}(x)=v^{t}(\varphi_{t,s}^{N}(x))
\end{align*}
where $v^{t}$ is given by $v^{t}(x)=(p,\bar{f}_{t}^{N}(q))$.
Here the mean-field force $\bar{f}_{t}^{N}$ is defined as $\bar{f}_{t}^{N}=f_N^{\beta}*\tilde{k}_{t}^{N}$ and $\tilde{k}_t^{N}:\mathbb{R}\times \mathbb{R}^{3}\rightarrow \mathbb{R}_{0}^{+}$ is given by 
\begin{align*}
\tilde{k}_{t}^{N}(q):=\int k_t^{N}(p,q) d^{3}p.
\end{align*}
\end{definition}
By using this approach, a new trajectory is obtained that is influenced by the mean-field force instead of the pair interaction force like in the Newtonian system \ref{Def:Newtonflow}.
Now we have two trajectories which we will compare and show later that they are close to each other.\\
To this end, we consider the lift of $\varphi^N_{t,s}(\cdot)$ to the $N$-particle phase-space, which we denote by ${}^N\Phi_{t,s}$. 
To lift the effective one-particle flow to the N-particle space we define the mean-field flow by:
\begin{definition}\label{Def:Mean-fieldflow}
The respective $\Phi_{t,s}^{N}=(\Phi_{t,s}^{1,N},\Phi_{t,s}^{2,N})=(\varphi_{t,s}^{N})^{\otimes N}$ satisfies 
\begin{align*}
\frac{\dd}{\dd t}\Phi_{t,s}^{N}(X)=\bar{V}_t(V),
\end{align*}
with $\bar{V}_t(X)=(P,\bar{F}_t(Q))$ and $\bar{F}_t$ given by $(\bar{F}_t(Q))_j:=\bar{f}_t^{N}(q_j)$.
\end{definition}

The mean-field particles move independent because the same force acts on every particle and we do not have pair interactions, which lead to correlations. 
In summary, for fixed $k_0$ and $N \in \mathbb{N}$, we consider for any initial configuration $X \in \mathbb{R}^{6N}$  two different time-evolutions: $\Psi_{t,0}^{N}(X)$, given by the microscopic equations and $\Phi_{t,0}^{N}(X)$, given by the time-dependent mean-field force generated by $f^N_t$. We are going to show that for typical $X$, the two time-evolutions are close in an appropriate sense. 
In other words, we have non-linear time-evolution in which $\varphi^N_{t,s}(\cdot\,; f_0)$ is the one-particle flow induced by the mean-field dynamics with initial distribution $k_0$, while, in turn, $k_0$ is transported with the flow $\varphi^N_{t,s}$. Due to the semi-group property $\varphi^N_{t,s'}\circ\varphi^N_{s',s} = \varphi^N_{t,s}$ it generally suffices to consider the initial time $s=0$.

In the following section we show that the two flows (\ref{Def:Mean-fieldflow}),(\ref{Def:Newtonflow}) are close to each other and so the microscopic and the macroscopic approach describe the same system. 
\section{A mean-field limit for the Vlasov system}
In the following section we show that the $N$-particle trajectory $\Psi_t$ starting from $\Psi_0$ (i.i.d. with the common density $k_0$) remains close to the mean-field trajectory $\Phi_t$ with the same initial configuration $\Psi_0=\Phi_0$ during any finite time $[0,T]$ and so the microscopic and the macroscopic approach describe the same system. Throughout this paper $C$ denotes a positive finite constant which may vary from place to pace but most importantly it will be independent of $N$.
\subsection{Statement of the results}
\begin{theorem} \label{maintheorem}
Let $T>0$ be such that a solution $k \in \mathcal{C}([0, T], \mathcal{H}^{2m-1}_{2r})$ of \eqref{Vlasov} exists. 
Moreover, let $(\Phi^{\infty}_{t,s})_{t, s\in \mathbb{R}}$ be the related lifted effective flow defined in \eqref{Def:Mean-fieldflow} as well as $({\Psi}^{N}_{t,s})_{t,s\in \mathbb{R}}$ the $N$-particle flow defined in \eqref{Def:Newtonflow}. 
 If  $0<\alpha<\beta<\frac{1}{7}$, then for any $\gamma>0$ there exists $C_\gamma>0$ such that for all $N\in \mathbb{N}$ with $N\geq N_0$ it holds that
\begin{align}
\mathbb{P}\big(X\in \mathbb{R}^{6N}:\sup_{0\le s \le T}|\Psi_{s,0}^{N}(X)-{\Phi_{s,0}^{\infty}}(X)|_{\infty}>  N^{-\alpha} \big)\le C_\gamma N^{-\gamma}.\label{result1}
\end{align}
\end{theorem}
It is intuitively clear that the coarse grained effective description gets more appropriate as the number of particles increases and becoming exact in the limit $N\rightarrow\infty$. This Theorem implies Propagation of Chaos an thus convergence of the marginals of the $N$-particle density towards products of solutions of the mean-field equation.

\subsection{Notation and preliminary studies}

The solution of the Vlasov-Dirac-Benney equation $k_t^{N}:\mathbb{R}^{6}\rightarrow\mathbb{R}_{0}^{+}$ can be understood as a one-particle probability density. All probabilities and expectation values are meant with respect to the product measure given at a certain time.
For any random variable $R:\mathbb{R}^{6N}\rightarrow\mathbb{R}$ and any element $B$ of the Borell $\sigma$-algebra we have
\begin{align*}
\mathbb{P}_t(R \in B)=\int_{{R}^{-1}(B)}\prod_{j=1}^{N}k^{N}(x_j)dX \hspace{0.3cm} \text{and}\hspace{0.3cm}
\mathbb{E}_t(R)=\int_{\mathbb{R}^{6N}}R(X)\prod_{j=1}^{N}k^{N}(x_j)dX.
\end{align*}
Since the measure is invariant under $\Phi_{t,s}^{N}$ it follows that 
\begin{align*}
\mathbb{E}_s(R\circ\Phi_{t,s}^{N})=\int_{\mathbb{R}^{6N}}R(\Phi_{t,s}^{N}(X))\prod_{j=1}^{N}k_s^{N}(x_j)dX=\int_{\mathbb{R}^{6N}}R(X)\prod_{j=1}^{N}k_s^{N}(\varphi_{s,t}^{N}(x_j))dX
\end{align*}
and since $k_s^{N}(\varphi_{s,t}^{N}(x_j))=k_t^{N}(x_j)$ we get $\mathbb{E}_s(R\circ\Phi_{t,s}^{N})=\mathbb{E}_t(R)$.
During the proof it is helpful to deviate into cases and therefore we will restrict our selves to certain configurations. For this purpose we introduce the restricted expectation value.

\begin{definition}
Let $Z$ be a random variable, $A \subset \mathbb{R}^{6N}$ a set,
then the restricted expectation value is given by
\begin{align*}
\mathbb{E}(Z|A):=\mathbb{E}(Z^{A}) \hspace{1cm} &\text{with} \hspace{1cm} Z^{A}(\omega)=     \begin{cases}
    Z(\omega)& \omega \in A\\
    0 & \omega \not \in A,\\
    \end{cases} \\
\mathbb{E}(Z)=\sum_{\omega\in \mathbb{R}^{6N}}Z(\omega)\mathbb{P}(\omega) \hspace{1cm}&\text{and} \hspace{1cm} \mathbb{E}(Z|A)=\sum_{\omega \in A}Z(\omega)\mathbb{P}(\omega).
\end{align*}
\end{definition}

Now we introduce a suitable notation of the distance on $\mathbb{R}^{6}$, which enables us to prove that for finite time $\Psi_{t,0}^{N}$ and $\Phi_{t,0}^{N}$ will typically be close with respect to that notation of this distance.
Since we are dealing with probabilistic initial conditions we introduce a stochastic process $J_t$ which is such that a small expectation value of $J_t$ implies that $\Psi_{t,0}^{N}(X)$ an $\Phi_{t,0}^{N}(X)$ are close as described  in the previous chapter.
In view of Theorem \ref{maintheorem}, our aim is to show that
$
\mathbb{P}_0(\sup_{0\leq s\leq t}|\Psi_{s,0}^{N}-\Phi_{s,0}^{N}|_{\infty}> N^{-\gamma})$ tends to zero faster than any inverse power of $N$. This will be implemented modifying the stochastic process $J_t$ \eqref{J_t heuristic} in order to separate the error terms coming form the law of large numbers from other sources of errors. This can be done by defining $J_t$ in the following way:
\begin{definition}\label{J_t}
Let $\Phi_{s,0}^{N}(X)$ be the mean-field flow defined in \eqref{Def:Mean-fieldflow} and $\Psi_{s,0}^{N}(X)$ the microscopic flow defined in \eqref{Def:Newtonflow}. We denote by $\Phi_{s,0}^{N,1}(X)=(q_i(t))_{1\leq i\leq N}$ and $\Phi_{s,0}^{N,2}(X)=(p_i(t))_{1\leq i\leq N}$ the projection onto the spatial or respectively the momentum coordinates.
Let for $T>0$ and without loose of generality $N>1$ the auxiliary process be defined as follows
\begin{align*}
&J_t(X):=\\
&\min\left\lbrace 1,  \sup_{0 \leq s \leq t}\left\lbrace \sigma_{N,t}N^{\alpha}\left( \  \sqrt{\ln(N)} \left| \Psi^{1,N}_{t,0}(X) - \Phi^{1,N}_{t,0}(X) \right|_\infty+ \left| \Psi^{2,N}_{t,0}(X) - \Phi^{2,N}_{t,0}(X)\right|_\infty+N^{5\beta-1} \right)	\right\rbrace  \right\rbrace
\end{align*}
for $0\leq t\leq T$ with scaling factor $\sigma_{N,t}= e^{\lambda \sqrt{\ln(N)}(T-s)}$.
Here $|\cdot|_{\infty}$ denotes the supremum norm on $\mathbb{R}^{6N}$.
\end{definition}
The metric $ \lvert ^N\Psi_{t,0}(X) - \Phi_{t,0}^{N}(X) \rvert_\infty$ is much stronger than usual weak distances between probability measures, thus allowing for better stability estimates.
The distance in spatial and momentum coordinates are weighted differently, exploiting the second-order nature of the system.

The spatial and momentum coordinates are weighed differently to take advantage of the system's second-order nature when comparing microscopic trajectories to the mean-field equation's characteristic curves. The growth of the spatial distance is trivially bounded by the difference of the respective momenta. The idea is thus to be a little more strict on deviations in space, so to speak, and use this to obtain better control on fluctuations of the force.
Moreover the scaling factor $e^{\lambda \sqrt{\ln(N)}(T-s)}$ optimizes the rate of convergence as compensates the time dependent natural fluctuations.

We now want to estimate the time derivative of $E(J_t)$.
Whenever the random variable $J_t$ has the value $1$, it has reached its maximum so the inequality $\partial_t^{+}\mathbb{E}(J_t)\leq C(\mathbb{E}(J_t)+o_N(1)))$ is trivial. The configurations where $J_t$ is maximal, that is $|\Psi_{t,0}^{N}(X)-\Phi_{t,0}^{N}(X)|_{\infty}\geq N^{-\alpha}$ are irrelevant for finding an upper bound of $\partial_t^{+}\mathbb{E}(J_t)$. The set of such configurations will be called $\mathcal{A}_{t}$.
We will show  that the expectations value $\partial_t^{+}\mathbb{E}(J_t)$ restricted on the set $\mathcal{A}_{t}$ is less or equal $0$.

We only have to consider the cases where $J_t$ is smaller than $1$ since $\partial_t^{+}J_t=0$ for $J_t=1$, but then we have the boundary condition by definition of the random variable, i.e $|\Psi_{t,0}^{2,N}(X)-\Phi_{t,0}^{2,N}(X)|_{\infty}< N^{-\alpha}$.
Due to the pre-factor $\sqrt{\ln(N)}$ in Definition \ref{J_t}, the particular anisotropic scaling of our metric will allow us to ``trade'' part of this  divergence for a tighter control on spatial fluctuations. This will suffice to establish the desired convergence, using the fact that $e^{\sqrt{\ln(N)}}$ grows slower than $N^\epsilon$ for any $\epsilon >0$. (\cite{Dustin} implemented the same idea).

By the definition of $J_t$ we get the boundary condition for free and
the construction of $J_t$ motivates the following lemma. 
\begin{lemma}\label{Mainlemma}
For $t>0$ there exists a constant $C_\gamma<\infty$, under the assumptions of Theorem \ref{maintheorem}, such that 
\begin{align*}
\mathbb{E}_0(J_t)\leq C_{\gamma}N^{-\gamma}.
\end{align*}
for $\gamma>0.$
\end{lemma}
Theorem \ref{maintheorem} follows directly from Lemma \ref{Mainlemma}, since the following probability can be estimated according to the description above.
\begin{align*}
\mathbb{P}_{0}(\sup_{0 \leq s\leq t}|\Psi_{s,0}^{N}(X)-\Phi_{0,s}^{N}(X)|)\geq N^{-\alpha})=\mathbb{P}_{0}(J_t=1)\leq \mathbb{E}_{0}(J_t).
\end{align*}
The proof of Lemma \ref{Mainlemma} is based on a Gronwall argument and therefore we will give an upper bound on  $\partial_t^{+}\mathbb{E}_0(J_{t})$ by introducing a suitable partition of the phase space $\mathbb{R}^{6N}$.
A first observation is that the growth of  $\partial_t^{+}\mathbb{E}_0(J_{t})$ stems from the fluctuation in the force, which itself can be estimated by
\begin{align*}
|F(\Psi_{t,0}^{N}(X))-\bar{F}(\Phi_{t,0}^{N}(X))|_{\infty} \leq |F(\Psi_{t,0}^{N}(X))-F(\Phi_{t,0}^{N}(X))|_{\infty}+|F(\Phi_{t,0}^{N}(X))-\bar{F}(\Phi_{t,0}^{N}(X))|.
\end{align*}
To control theses two addends we will introduce unlikely sets.
The sets of configurations $X$ for which the second term $|F(\Phi_{t,0}^{N}(X))-\bar{F}(\Phi_{t,0}^{N}(X))|_{\infty}$ is large will be denoted by $\mathcal{B}_t$. Large means in our case larger than $N^{5\beta-1}\ln(N)$.

Since any difference in the force is directly translatable into a growth in the difference $|\Psi_{t,0}^{N}(X)-\Phi_{t,0}^{N}(X)|_{\infty}$ which is multiplied by $N^{-\alpha}$ in the definition of $J_t$.
We will see, that the probability to be in $\mathcal{B}_t$ is indeed small.
If $F$ was globally Lipschitz continuous, the first term  $|F(\Psi_{t,0}^{N}(X))-F(\Phi_{t,0}^{N}(X))|_{\infty}$ would directly translate in the difference $|\Psi_{t,0}^{N}(X)-\Phi_{t,0}^{N}(X)|_{\infty}$ as
$
|F(\Psi_{t,0}^{N}(X))-F(\Phi_{t,0}^{N}(X))|_{\infty}\leq L_{Lip} |\Psi_{t,0}^{N}(X)-\Phi_{t,0}^{N}(X)|_{\infty}
$
and the result would be proven.
The forces we consider are singular and so unfortunately there exists no global Lipschitz constant. Although there exist configurations for which the force becomes singular in the limit $N\rightarrow \infty$, for example when all particles have the same position, these configurations are not very likely.
To control the first addend and to implement this argument we will introduce a function $g$ in Definition \ref{force g} which controls the difference $|f^{N}(x)-f^{N}(x+\delta)|_{\infty}$ for a $ 2 N^{-\alpha}>\delta \in \mathbb{R}^{3}$. This will be proven in Lemma \ref{lemma Definition G} observing that  we only need to take into account fluctuations smaller than $N^{-\alpha}$ by Definition of $J_t$. 
In a further step we will control $G=\frac{1}{N}\sum_{j=1}^{N}g(q_j-q_k)$ for typical configurations and proof that the set where $G$ is large, denoted by $C_t$, is very unlikely. For the configurations which are left we use the fact that the force term is compactly supported and the associated scaling behaviour is in our favour. Consequently we will get a good estimate on $|F(\Psi_{t,0}^{N}(X))-F(\Phi_{t,0}^{N}(X))|_{\infty}$.

\subsubsection{Controlling the growth of the force}
In the following section we overcome the problem that forces $f_N^{\beta}$ become singular in the limit $N \to \infty$ and hence do not satisfy a uniform Lipschitz bound.
   The function $l:\mathbb{R}^3\rightarrow\mathbb{R}$ occurring in the definition of $f_N$ was defined such that $D_1 l,D_2l$ and $D_3l$ are defined everywhere and are bounded functions.
   Its easy to check that $l^N(q)$ satisfy a Lipschitz condition by a mean value argument.
   \begin{lemma}\label{l lipschitz}
Let $l\in L^{\infty}(\mathbb{R}^3)\cap L^1(\mathbb{R}^3): ||D^{\alpha}l||_{\infty}\leq C_{\alpha}$ be a smooth vanishing at infinity function with $\alpha=( \alpha_1,\hdots,\alpha_n)\in\mathbb{N}^{n}_0$. Then there exists a $L>0$ such that
   \begin{align*}
   |l(a)-l(b)|_{\infty}\leq L |a-b|_{\infty}
   \end{align*}
   \end{lemma}
   \begin{proof}
Since $D_il$ is bounded for $i\in \lbrace 1,2,3 \rbrace$
there exists $ S_1=\sup(||D_1f(x)||:X\in \mathbb{R}^3)
, S_2=\sup(||D_2l(x)||:X\in \mathbb{R}^3)$ and $S_3=\sup(||D_1l(x)||:X\in \mathbb{R}^3)$.
For $a,b\in \mathbb{R}^3$ we can estimate the difference $l(a)-l(b)$ by the triangle inequality
\begin{align}\label{ab Dreieck}
|l(a)-l(b)|\leq |l(a_1,a_2,a_3)-l(a_1,a_2,b_3)|+|l(a_1,a_2,b_3)-l(a_1,b_2,b_3)|+|l(a_1,b_2,b_3)-l(
b_1,b_2,b_3)|.
\end{align}
Since the partial derivatives exist everywhere in $\R^{3}$, we can use the one-dimensional mean Value Theorem to show that there exists some $\xi$ such that: 
\begin{align*}
\frac{l(a_1,a_2,a_3)-l(a_1,b_2,a_3)}{a_2-b_2}=D_2l(a_1,\xi,a_3).
\end{align*}
By the definition of $S_2$, it follows that $|l(a_1,a_2,a_3)-l(a_1,b_2,a_3)|\leq S_2 |a_2-b_2|.$
And similarly for the other addends of \eqref{ab Dreieck}.
Additionally using Cauchy-Schwarz inequality, we obtain 
\begin{align*}
|l(a)-l(b)|&\leq S_1 |a_1-b_1|+S_2|a_2-b_2|+S_3|a_3-b_3|\\
&\leq \sqrt{S_1^2+S_2^2+S_3^2}\cdot((a_1-b_1)^2+(a_2-b_2)^2+(a_3-b_3)^2)\\
&=\sqrt{S_1^2+S_2^2+S_3^2}||a-b||
\end{align*}
So $l$ is indeed Lipschitz continuous with $L=\sqrt{S_1^2+S_2^2+S_3^2}$.
   \end{proof}
   Next we define a function $g_N^\beta$, which provides a bound for fluctuations of $f_N^\beta$.
\begin{definition}\label{force g}
Let $g:\R^3\rightarrow\R^3$ with
\begin{align*}
g_N^\beta(q):= L\cdot N^{5\beta} \mathds 1_{\{\text{supp}_l\}}(N^{\beta}q).
\end{align*}
and the total fluctuation $G$ be defined by $(G(X))_j:=\sum_{j=1}^{N}\frac{1}{N}g_N^\beta(q_i-q_j)$. Furthermore $\bar{G}_t$ is given by $(\bar{G}_t(X))_j:=\bar{g}_t(q_j)$ with $\bar{g}(q)=g_N^\beta\ast\tilde{k}_t^{N}(q)$.
\end{definition}
To show that the difference $|f_N^\beta(x)-f_N^\beta(x+\delta)|_{\infty}$ can be controlled by $g_N^\beta$ we prove the following lemma.
\begin{lemma}\label{lemma Definition G}
For any $\delta \in \mathbb{R}^{3}$ it follows that 
\begin{align*}
|f_N^\beta(x)-f_N^\beta(x+\delta)|_{\infty}\leq g_N^\beta(q)|\delta|_{\infty}.
\end{align*}
\end{lemma}
\begin{proof}
We recall that $f_N^\beta$ was defined by
\begin{align*}
f_N^\beta(q)= N^{4\beta}l(N^{\beta}q)
\end{align*}
and that $l$ is Lipschitz continuous by Lemma \ref{l lipschitz}. Hence we get
\begin{align*}
|f_N^\beta(x)-f_N^\beta(x+\delta)|_{\infty}\leq L N^{4\beta}|N^{\beta}\delta|_{\infty}\leq LN^{5\beta}|\delta|_{\infty}=g_N^\beta(q)|\delta|_{\infty}.
\end{align*}
\end{proof}
Notice that in the case where we use $G$ to control the fluctuation we know by the construction of $J_t$ that $|\Psi-\Phi|<N^{-\alpha}$.
Furthermore the following observations of $f_N^\beta$ and $g_N^\beta$ turn out to be very helpful in the sequel.
One crucial consequence of the bounded density is that the mean-field
force remains bounded, as well.
\begin{lemma}\label{Abschätzungen Faltung}
	Let $g^N(x)$ be defined in Definition \ref{force g} and $\tilde{k} \in W^{2,1}(\R^3)\cap W^{2,\infty}(\R^3)$.  Then there exists a constant $C>0$ independent of $N$ such that
	\begin{align}\label{con1}
	\norm{f_N^\beta\ast \tilde{k}}_\infty \leq C\norm{\nabla\tilde{k}}_\infty
	\end{align}
		and
		\begin{align}\label{con2}
	\quad	\norm{g_N^\beta\ast \tilde{k}}_\infty\leq C\norm{\Delta\tilde{k}}_\infty.
		\end{align}
\end{lemma}
\begin{proof}
The function $\phi:\mathbb{R}^3\rightarrow\mathbb{R}$ is rotationally symmetric. It holds that $\phi(x)=h(||x||)$ for some  measurable functions $h:\mathbb{R}_0^+\rightarrow\mathbb{R}$ and hence
$\int_{K_{a,b}}f(x)dx=\omega_n \int_{a}^{b} h(r)r^{2}dr$ for $K_{a,b}:=\lbrace X\in \R^3:a<||x||<b\rbrace$ .
For $\phi\in C^{\infty}_0(\mathbb{R}^3)$ the theorem on coordinate transformation provides
\begin{align*}
\norm{\phi^N}_1=\int_{\mathbb{R}^3}|N^{3\beta}\phi(N^{\beta}x) dx|\leq C \int_{\mathbb{R}^3}|\phi(y) dy|\leq C
\end{align*}
and for $f_N^\beta\in C^\infty(\mathbb{R}^3)$ Leibniz integral rule implies
\begin{align*}
\norm{f_N^\beta\ast \tilde{k}}_\infty =\norm{\nabla \phi^N\ast \tilde{k}}_\infty=\norm{ \phi^N\ast \nabla\tilde{k}}_\infty\leq C \norm{\phi^N}_1\norm{\nabla \tilde{k}}_\infty\leq C\norm{\nabla \tilde{k}}_\infty.
\end{align*} 
Analogously we can estimate
\begin{align*}
\norm{g_N^\beta\ast \tilde{k}}_\infty =\norm{\Delta \phi^N\ast \tilde{k}}_\infty=\norm{ \phi^N\ast \Delta\tilde{k}}_\infty\leq C \norm{\phi^N}_1\norm{\Delta \tilde{k}}_\infty\leq C\norm{\Delta \tilde{k}}_\infty.
\end{align*} 
By applying the Vlasov property $k_t(q,p)=k_0(\bar{X}(q,p))$ and the assumptions on $k$ according to the solution theory we can see that $\norm{f^N\ast \tilde{k}}_\infty $ and $\norm{g^N\ast \tilde{k}}_\infty $ are bounded by a constant not depending on $N$.
Alternatively one can estimate by using the Maclaurin series of $\tilde{k}\in W^{2,1}(\mathbb{R}^3)$ and the theorem on coordinate transform
		\begin{align*}
		\norm{g_N^\beta\ast \tilde{k}}_\infty&=\norm{ \Delta\phi^N\ast T_n \tilde{k}(x; 0))}_\infty=\norm{ \nabla\nabla\phi^N\ast T_n \tilde{k}(x; 0))}_\infty=\norm{ \nabla\phi^N\ast\nabla T_n \tilde{k}(x; 0))}_\infty\\
		&=\norm{ f_N^\beta\ast \nabla\sum_{|\alpha| = 0}^{n}\frac{x^{\alpha}}{\alpha !} D^{\alpha}\tilde{k}(0)}_\infty\leq\norm{f^N \ast \nabla(T_2 \tilde{k}(x; 0) + \mathcal{O}(|x|^2))}\\
		&\leq\norm{\int_{\mathbb{R}^3}f_N^{\beta}(x-y)(T_1\tilde{k}(y,0)+\mathcal{O}(|y|^1))dy}\\
		&\leq\norm{\int_{\mathbb{R}^3}N^{4\beta}l^N(N^{\beta}(x-y))(T_1\tilde{k}(y,0)+\mathcal{O}(|y|^1))dy}\\
&\leq	 C+C\mathcal{O}(N^{-2\beta})
		\end{align*}	
		In the last step we used the symmetry of $f^N$ integration by substitution and the fact that $f^N$ is compactly supported.
\end{proof}

\subsection{The evolution of $E(J_t)$}
 Since $\frac{\dd}{\dd t} J^{N}_t(X) \leq 0$ if $ \sup\limits_{0 \leq s \leq t} \lvert ^N\Psi_{s,0}(X) - {}^N\Phi_{s,0}(X) \rvert_\infty \geq N^{-\alpha}$ we only have to consider situations in which mean-field trajectories and microscopic trajectories are close.
In order to control the evolution of $E(J_t)$ 
we will partition the phase space as described in Section \ref{Heuristics}.
\begin{definition}\label{definition ABC}
Let for any $t\in \mathbb{R}$ the sets $\mathcal{A}_t,  \mathcal{B}_t, \mathcal{C}_t$ be given by
\begin{align*}
X\in \mathcal{A}_t &\Leftrightarrow |J_t|=1 \\
X\in \mathcal{B}_t &\Leftrightarrow |F(\Phi_{t,0}^{N}(X))-\bar{F}(\Phi_{t,0}^{N}(X))|_{\infty}>N^{-1+5\beta}\ln(N)\\
X\in \mathcal{C}_t &\Leftrightarrow |G(\Phi_{t,0}^{N}(X))-\bar{G}(\Phi_{t,0}^{N}(X))|_{\infty}>N^{-1+7\beta}\ln(N).
\end{align*}
\end{definition}
For estimating the probability of configurations $X\in\mathcal{B}_t$ and $X\in\mathcal{C}_t$ we will use a law of large numbers argument. It will turn out, that these configurations are very unlikely.

\subsubsection{Law of large numbers}
The method is designed for stochastic initial conditions, thus allowing for law of large number estimates that turn out to be very powerful. Note that the particles evolving with the mean-field flow remain statistically independent at all times.
We use the following Lemma to provide the probability bounds of random variables.
\begin{lemma}\label{Wahrscheinlichkeitslemma} 
	Let $Z_1,\cdots,Z_N$ be $i.i.d.$ random variables with $\mathbb{E}[Z_i]=0,$ $\mathbb{E}[Z_i^2]\leq r(N)$
	and $|Z_i|\leq C\sqrt{Nr(N)}$. Then for any $\alpha>0$, the sample mean $\bar{Z}=\frac{1}{N}\sum_{i=1}^{N}Z_i$ satisfies
	\begin{align*}
	\mathbb{P}\left(|\bar{Z}|\geq\frac{C_\gamma \sqrt{r(N)}\ln(N)}{\sqrt{N}}\right)\leq N^{-\gamma},
\end{align*}
	where $C_\gamma$ depends only on $C$ and $\gamma$.
\end{lemma}
The proof can be seen in \cite[Lemma 1]{Goodman}. It is a direct result of  Taylor's expansion and Markov's inequality.
Furthermore it is a direct consequence of the following Lemma.
\begin{lemma}\label{Erwartungswertlemma}
For $N\in\mathbb{N}$ let $Z_1,\hdots,Z_N$ be  independent and identically distributed random variables on $\R^3$ with $||Z_j||_{\infty}\leq C $, $\mathbb{E}(Z_j)=0$, $\mathbb{E}(Z_j^2)\leq \frac{C}{N}$ for all $i\in\lbrace 1,\hdots, N\rbrace$, then the finite sum of the random variables $S_N:=\sum_{i=1}^{N}Z_i$ full fills
\begin{align*}
\mathbb{E}(e^{|S_N|})\leq C.
\end{align*}
\end{lemma}
\begin{proof}
From the Taylor series expansion we have $e^x= 1+x+x^2e^{\mu}$  with $|\mu|<|x|$. As the expectation value is linear and using the properties of the random variable we get
\begin{align*}
\mathbb{E}(e^Z)=1+\mathbb{E}(X)+\mathbb{E}(Z^2\cdot e^{\mu}) \leq 1+0+\mathbb{E}(Z^2)\cdot C
\leq 1+\frac{C}{N}.
\end{align*}
For the (positive or negative) sum of the independent and identically distributed random variables a similar inequality follows 
\begin{align*}
\mathbb{E}(e^{\pm S_n})=\mathbb{E}(e^{\pm\sum_{j=1}^{N}Z_j})\overset{\text{iid}}{=}(\mathbb{E}(e^{\pm Z_j}))^N\leq (1\pm\frac{C}{N})^N \overset{N\rightarrow\infty}{\longrightarrow} e^{\pm C}.
\end{align*}
So in total we get $\mathbb{E}(e^{|S_n|})\leq C$ as $e^{|x|}\leq e^x+e^{-x}$ for all $X\in \R$.
\end{proof}
By Markov inequality we get for $\mu>0$
\begin{align*}
&\mathbb{P}(e^{|S_N|}\geq\mu)\leq\frac{C}{\mu}
\Rightarrow \mathbb{P}(|S_N|\geq \ln(\mu))\leq\frac{C}{\ln(\mu)}
\end{align*}
and Lemma \ref{Wahrscheinlichkeitslemma} is a direct consequence.\\
Now we will estimate the probability of the unlikely sets defined in Definition \ref{definition ABC}. Therefore we recall the notation
\begin{align}\label{barF,barG}
(\overline F^N(\overline X_t))_i:=\int_{\R^3} f_N^{\beta}(\overline x_i^t-x)k^N(x,t)dx \ \ \text{and} \ \ (\overline G^N(\overline X_t))_i:=\int_{\R^3} g_N^{\beta}(\overline x_i^t-x)k^N(x,t)dx
\end{align}
and introduce the underlying version of the Law of Large Numbers for the paper.
\begin{lemma}\label{lmlarge} At any fixed time $t\in[0,T]$, suppose that $\overline X_t$ satisfies the mean-field dynamics \eqref{Def:Mean-fieldflow}, $F^N$ and $\overline F^N$ are defined in \eqref{force f} and \eqref{barF,barG} respectively,  $G^N$ and $\overline G^N$ are introduced in Definition \ref{lemma Definition G}. For any $\gamma>0$ and $0\leq\beta<\frac{1}{7}$, there exist a constant $C_{\gamma}>0$ depending only on $\gamma$, $T$ and $k_0$ such that
	\begin{equation}\label{F prob bound}
	\mathbb{P}\left(\left\| F^N(\overline X_t)-\overline F^N(\overline X_t)\right\|_\infty\geq C_{\gamma} N^{5\beta-1}\ln(N)\right)\leq N^{-\gamma},
	\end{equation}
	and
	\begin{equation}\label{G prob bound}
	\mathbb{P}\left(\left\| G^N(\overline X_{t})-\overline{G}^N(\overline X_{t})\right\|_\infty\geq C_{\gamma} N^{7\beta-1}\ln(N)\right)\leq N^{-\gamma}.
	\end{equation}
\end{lemma}
\begin{proof}
	We can prove the present lemma by using Lemma \ref{Wahrscheinlichkeitslemma} and the following  generalized version of the Young's inequality for convolutions for $p,r\in L^{1}$.
\begin{align*}
\|p\ast r\|_\infty\leq& \|p\|_{1}\|r\|_{\infty}
\end{align*}
Since $\|\tilde{k}^N_t\|_1 = 1$ and $\|\tilde{k}^N_t\|_\infty$ is bounded it holds due to Lemma \ref{Abschätzungen Faltung} that 
\begin{align}
&\|\tilde{k}^N_t(q_1)\ast f_N^\beta\|_\infty\leq\|\nabla\tilde{k}^N_t(q_1)\|_{1}
\leq C \|\nabla\tilde{k}\|_{1}\label{inequ1} \\ 
&\|\tilde{k}^N_t(q_1)\ast g_N^\beta\|_\infty\leq C\|\Delta\tilde{k}^N_t(q_1)\|_{1}.\label{inequ2}
\end{align}
Hence we get
\begin{align*}&\left\|f_N^\beta\ \ast \tilde{k}^N_t(q_1)\right\|_\infty\leq C\hspace{1cm}\text{and}\hspace{1cm}
\left\|g_N^\beta\ \ast \tilde{k}^N_t(q_1)\right\|_\infty\leq C.
\end{align*}
Using inequality \eqref{inequ1} and \eqref{inequ2} we have $\left|f_N^\beta(q_1-q_j)-f_N^\beta\ast \tilde{k}^N_t(q_1)\right|\leq C $ and analogously in the case of the fluctuation we have  $\left|g_N^\beta(q_1-q_j)-g_N^\beta\ast \tilde{k}^N_t(q_1)\right|\leq C$. 
The expectation values of $\left|f_N^{\beta}(q_1-q_j)\right|^2$ and $\left|g_N^{\beta}(q_1-q_j)\right|^2$ can be estimated using the theorem on coordinate transformation  by
\begin{align*}
\int \tilde{k}^N_t(q_j)\left|f_N^{\beta}(q_1-q_j)\right|^2d^3q_j\leq C\int \Big(N^{4\beta}l(N^{\beta}q)\Big)^2 d^3q
\leq CN^{8\beta}N^{-3\beta}
\leq C N^{5\beta}
\;
\end{align*}
and analogously
\begin{align*}
&\int \tilde{k}^N_t(q_j)\left|g_N^{\beta}(q_1-q_j)\right|^2d^3q_j\leq C\int \Big(CN^{5\beta}\mathds 1_{\{\text{supp}_l\}}(N^{\beta}q)\Big)^2 d^3q
\leq CN^{10\beta-3\beta} \leq C N^{7\beta}.
\;
\end{align*}
	Due to the exchangeability of the particles, we can estimate 
	\begin{align*}
	(F^N(\overline X_t))_1-(\overline F^N(\overline X_t))_1=\frac{1}{N}\sum_{j=2}^Nf_N^{\beta}(\overline q_1^t-\overline q_j^t)-\int_{\R^3} f_N^{\beta}(\overline q_1^t-q)k^N(q,t)dq=\frac{1}{N}\sum_{j=2}^{N}Z_j,
	\end{align*}
for the random variable $$Z_j:=f_N^{\beta}(\overline q_1^t-\overline q_j^t)-\int_{\R^3} f_N^{\beta}(\overline q_1^t-q)k^N(q,t)dq.$$ 
	Since $\overline q_1^t$ and $\overline q_j^t$ are independent when $j\neq 1$ and $f_N^{\beta}(0)=0$, let us consider $\overline q_1^t$ as given and denote $\mathbb{E'}[\cdot]=\mathbb{E}[\cdot|\overline q_1^t]$ and the condition 
	$\mathbb{E}'[Z_j]=0$ of Lemma \ref{Wahrscheinlichkeitslemma} holds since
	\begin{align*}
	\mathbb{E}'\left[f_N^{\beta}(\overline q_1^t-\overline q_j^t)\right]=\int_{\R^6} f_N^{\beta}(\overline q_1^t-q)k^N(q,p,t)dxdp\notag 
	=\int_{\R^3} k^N(\overline q_1^t-q)k^N(q,t)dq.
	\end{align*}
	Additionally we need a bound for the variance as well
	\begin{align*}
	\mathbb{E}'\big[|Z_j|^2\big]=\mathbb{E}'\left[\left|k^N(\overline q_1^t-\overline q_j^t)-\int_{\R^3} k^N(\overline q_1^t-q)\rho^N(q,t)dq\right|^2\right].
	\end{align*}
	We know that
	\begin{align*}
\int_{\R^3} f_N^{\beta}(\overline q_1^t-q)k^N(q,t)dq\leq C(\norm{k^N}_1+\norm{k^N}_\infty),
	\end{align*}
which suffices to estimate
	\begin{align*}
	\mathbb{E'}\big[f_N^{\beta}(\overline q_1^t-\overline q_j^t)\big]=\int_{\R^3} f_N^{\beta}(\overline q_1^t-q)k^N(q,t)dq\leq C(\norm{k^N}_1+\norm{k^N}_\infty)N^{\beta}\leq C,
	\end{align*}
	and additionally the variance
	\begin{align*}
	\mathbb{E'}\big[f_N^{\beta}(\overline q_1^t-\overline q_j^t)^2\big]=\int_{\R^3} f_N^{\beta}(\overline q_1^t-q)^2k^N(q,t)dq\leq \norm{k^N}_\infty\norm{k^N}_2^2N^{5\beta}\leq CN^{5\beta}.
	\end{align*}
	
	Hence one has
	\begin{align*}
	\mathbb{E}'\big[|Z_j|^2\big]\leq CN^{5\beta}.
	\end{align*}
It follows, for $r(N)=CN^{5\beta}$, that $|Z_j|\leq C\sqrt{Nr(N)}$. Hence, using Lemma \ref{Wahrscheinlichkeitslemma}, we have the probability bound 
	\begin{align*}
	\mathbb{P}\left(\left|(F^N(\overline X_t))_1-(\overline F^N(\overline X_t))_1\right|\geq C N^{5\beta-1}\ln(N)\right)\leq N^{-\gamma}.
	\end{align*}
	Similarly, the same bound  also  holds for all other indexes  $i=2,\cdots,N$, which leads to
	\begin{align}\label{F all}
	\mathbb{P}\left(\left\| F^N(\overline X_t)-\overline F^N(\overline X_t)\right\|_\infty\geq C N^{5\beta-1}\ln(N)\right)\leq N^{1-\gamma}.
	\end{align}
	Let $C_{\gamma}$ be the constant $C$ in  \eqref{F all} which is only depending on $\gamma,T$ and $k_0$, then we conclude \eqref{F prob bound}.
	
	To prove \eqref{G prob bound}, we follow the same procedure as above
	\begin{align*}
	(G^N(\overline X_t))_1-(\overline G^N(\overline X_t))_1=\frac{1}{N}\sum_{j=2}^Ng_N^{\beta}(\overline q_1^t-\overline q_j^t)-\int_{\R^3} g_N^{\beta}(\overline q_1^t-q)k^N(q,t)dq=\frac{1}{N}\sum_{j=2}^{N}Z_j,
	\end{align*}
	with the random variable $$Z_j=g_N^{\beta}(\overline q_1^t-\overline q_j^t)-\int_{\R^3} g_N^{\beta}(\overline q_1^t-q)k^N(q,t)dq.$$ 
It holds that $\mathbb{E}'[Z_j]=0$ and
	for the bound for the variance we compute that
	\begin{align*}
	\mathbb{E'}\big[g_N^{\beta}(\overline q_1^t-\overline q_j^t)\big]=\int_{\R^3} g_N^{\beta}(\overline q_1^t-q)k^N(q,t)dq\leq N^{2\beta} C(\norm{k}_1+\norm{k}_\infty)\leq C,
	\end{align*}
	and
	\begin{align*}
	\mathbb{E'}\big[g^N(\overline q_1^t-\overline q_j^t)^2\big]=\int_{\R^3} g_N^{\beta}(\overline q_1^t-q)^2k^N(q,t)dq\leq CN^{7\beta}(\norm{k}_1+\norm{k}_\infty)\leq CN^{7\beta},
	\end{align*}
	where we have used the definition of $g_N^{\beta}$ and the fact that by integration by substitution we get the rescaling factor $N^{-3\beta}$ in dimension in $d=3$.
	Hence one has
	\begin{align*}
	\mathbb{E}'\big[|Z_j|^2\big]\leq CN^{7\beta}.
	\end{align*}
	
	So the hypotheses of Lemma \ref{Wahrscheinlichkeitslemma} are satisfied with $r(N)=CN^{7\beta}$ and additional we can estimate $|Z_j|\leq C\sqrt{Nr(N)}$. Hence, we have the probability bound 
	\begin{align*}
	\mathbb{P}\left(\left|(G^N(\overline q_t))_1-(\overline G^N(\overline q_t))_1\right|\geq C N^{7\beta-1}\ln(N)\right)\leq N^{-\gamma},
	\end{align*}
	by Lemma \ref{Wahrscheinlichkeitslemma}, which leads to
	\begin{align}\label{G all}
	\mathbb{P}\left(\left\| G^N(\overline X_t)-\overline G^N(\overline X_t)\right\|_\infty\geq C N^{7\beta-1}\ln(N)\right)\leq N^{1-\gamma}.
	\end{align}
Thus, \eqref{G prob bound} follows from \eqref{G all}.	
\end{proof}
So far we could show, that the probability  of $X$ being in one of the unlikely set defined in \ref{definition ABC} decreases faster than any negative power of $N$.
For any time $0<t<T$, initial conditions in $(B_t\cap C_t)^c$ are typical with respect to the product measure $K_0 := \otimes^{N}
k_0$ on $\R^{6N}.$
\subsubsection{Controlling the Expectation value of $J_t$}
We are left to estimate the expectation $E_0(J^{N}_t)$ and remember that it was split into
	\begin{equation*} E_0(J^{N}_t) = \E_0(J^{N}_t \mid \mathcal{A}_t) + \E_0(J^{N}_t \mid  \mathcal{A}_t^c\setminus(\mathcal{B}_t^c\cap C_t^c) + \E_0(J^{N}_t \mid (\mathcal{A}_t\cup\mathcal{B}_t\cup\mathcal{C}_t)^c. \end{equation*}
 As we already know that, on the set $\mathcal{A}_t$ the process $J^{N}_t(X)$ is already maximal and we have $\frac{\mathrm{d}}{\mathrm{d}t} J^{N}_t(X) =0 $ and thus also
$ \frac{\mathrm{d}}{\mathrm{d}t}\, \E_t(J^{N}_t \mid \mathcal{A}_t) =0$.
 To estimate the remaining terms we remember that for $X \in \mathcal{A}_{t}^c$ the probability for $X \in  \mathcal{B}_t \cap \mathcal{C}_t$ decreases faster than any power of $N$. 
Further more right derivative of $J_t$ with respect to $t$ is given by
\begin{align*}
	&\partial_t^+ J_t^{N, \lambda}(X)\\
	 &\leq \max \Bigl\lbrace 0 , \frac{d}{dt}\left(\sigma_{N,t}\left(N^{\alpha} \sqrt{\ln(N)} \left| ^N\Psi^1_{t,0}(X) - \Phi^{1,N}_{t,0}(X) \right|_\infty
	 +N^\alpha \left| \Psi^{2,N}_{t,0}(X) - \Phi^{2,N}_{t,0}(X)\right|_\infty+N^{5\beta+\alpha -1}
	\right) \right)  \Bigr \rbrace\\
	&\leq\max \Bigl\lbrace 0 , -\lambda \sqrt{\ln(N)} e^{\lambda \sqrt{\ln(N)}(T-t)}(N^{\alpha} \sqrt{\ln(N)} \lvert \ ^N\Psi^1_{t,0}(X) - {}^N\Phi^1_{t,0}(X) \rvert_\infty\\&\  \ \hspace{5,5cm} +\lvert ^N\Psi^2_{t,0}(X) - {}^N\Phi^2_{t,0}(X) \rvert_\infty+N^{5\beta+\alpha -1})\\
	&\ \ \ \ \ \ \ \ \ \ \ \ \ \ +  e^{\lambda \sqrt{\ln(N)}(T-t)} N^{\alpha} \partial_t \left(\sqrt{\ln(N)} \lvert \ ^N\Psi^1_{t,0}(X) - {}^N\Phi^1_{t,0}(X) \rvert_\infty+\lvert ^N\Psi^2_{t,0}(X) - {}^N\Phi^2_{t,0}(X) \rvert_\infty\right)
	 \Bigr \rbrace	
	\end{align*}
	For the derivative of the position coordinate and for the momentum coordinate we further estimate
\begin{align*}
\partial_t \lvert {}^N\Psi^1 _{t,0}(X) - \Phi^{1,N}_{t,0}(X) \rvert_\infty &\leq \lvert \partial_t (\Psi^{1,N}_{t,0}(X) - \Phi^{1,N}_{t,0}(X)) \rvert_\infty\\\notag
&\leq \lvert {}^N\Psi^2 _{t,0}(X) - \Phi^{2,N}_{t,0}(X) \rvert_\infty \leq \sup\limits_{0 \leq s \leq t} \lvert {}^N\Psi^2 _{s,0}(X) - {}^N\Phi^2_{s,0}(X) \rvert_\infty \\
\partial_t \lvert {}^N\Psi^2 _{t,0}(X) - \Phi^{2,N}_{t,0}(X) \rvert_\infty &\leq  \lvert \partial_t ({}^N\Psi^2 _{t,0}(X) - \Phi^{2,N}_{t,0}(X) )\rvert_\infty
\leq  \lvert F(\Psi^1_{t,0}(X)) - \overline{F}_t(\Phi^1_{t,0}(X)) \rvert_\infty.
	\end{align*} 	
	
	Secondly the total force is bounded $|F(X)|_{\infty}\leq N^{4\beta}$ and the mean-field force $\bar{F}$ is of order one. Since $X\in \mathcal{A}_t^c$ we get $N^{\beta}|^N\Psi^2 _{t,0}(X) - \Phi^{2,N}_{t,0}(X) |\leq 1$ and
hence $\sup\lbrace \lvert \partial^+_t J^{N}_t(X) \rvert : X \in \mathcal{A}_{t}^{c} \rbrace \leq C e^{\lambda \sqrt{\ln(N)}T} N^{4\beta}$ for some $C >0$.
According to Lemma \ref{Wahrscheinlichkeitslemma}
the probability for $X\in \mathcal{B}_t\cap\mathcal{C}_t$	decreases faster than any power of $N$.
 Hence, we can find for any $\gamma > 0$ a constant $C_\gamma$, such that
\begin{align*}
\begin{split}\label{controlatypical} \partial_t^+ \mathbb{E}(J^{N}_t \mid \mathcal{A}_{t}^{c}\setminus(\mathcal{B}_{t}^{c}\cap\mathcal{C}_{t}^{c}))\leq \sup\lbrace \lvert \partial^+_t J^{N}_t(X) \rvert : X \in \mathcal{A}_{t}^{c} \rbrace\, \mathbb{P}\bigl[(\mathcal{A}_{t}\cup \mathcal{B}_{t})\bigr]\leq e^{\lambda \sqrt{\ln(N)}T} C_\gamma N^{-\gamma}. \end{split}
\end{align*}
It remains to control $E_0(J^{N}_t \mid (\mathcal{A}_t\cup\mathcal{B}_t\cup\mathcal{C}_t)^c)$ which is defined on the most likely initial conditions.
The relevant term can be estimated by
\begin{align*}\lvert F^N(\Psi^1_{t,0}(X)) - \overline{F}_t(\Phi^1_{t,0}(X)) \rvert_\infty \leq \lvert F^N(\Psi^1_{t,0}(X)) - F^N(\Phi^1_{t,0}(X)) \rvert_\infty +  \lvert F^N(\Phi^1_{t,0}(X)) - \overline{F}_t(\Phi^1_{t,0}(X)) \rvert_\infty. 
\end{align*}

 Since $X \notin \mathcal{B}_t$, it follows for the second addend 
	\begin{align*}\lvert F(\Phi^1_{t,0}(X)) - \overline{F}(\Phi^1_{t,0}(X)) \rvert_\infty < N^{5\beta-1}\ln(N).
	\end{align*} 
For the first addend we use the triangle inequality to get for any $1 \leq i \leq N$ 
	\begin{align*} \Bigl \lvert \bigl(F(\Psi^1_{t,0}(X)) - F(\Phi^1_{t,0}(X)) \bigr)_i \Bigr\rvert_\infty \leq \Bigl \lvert \frac{1}{N}\sum\limits_{j=1}^N f_N^{\beta}(\Psi^1_i - \Psi^1_j) - f_N^{\beta}(\Phi^1_i - \Phi^1_j) \Bigr\rvert_\infty \\ \notag
	\leq \frac{1}{N} \sum\limits_{j=1}^N \bigl \lvert f_N^{\beta}(\Psi^1_i - \Psi^1_j) - f_N^{\beta}(\Phi^1_i - \Phi^1_j) \bigr\rvert_\infty
	\end{align*}
and application of a version of mean value theorem stated in Lemma \ref{lemma Definition G} leads to
	\begin{align*}
	\bigl \lvert f_N^{\beta}(\Psi^1_i - \Psi^1_j) - f_N^{\beta}(\Phi^1_i - \Phi^1_j) \bigr\rvert_\infty &\leq g_N^{\beta}(\Phi^1_i - \Phi^1_j) \lvert (\Psi^1_i - \Psi^1_j) - (\Phi^1_i - \Phi^1_j) \rvert_\infty\\[1.2ex]
	&\leq 2\, g_N^{\beta}(\Phi^1_i - \Phi^1_j) \lvert \Psi^1_{t,0} - \Phi^1_{t,0} \rvert_\infty.
	\end{align*}
Since $X \in \mathcal{A}^c_t$ we have, by the construction of $J^{N}_t(X)$, in particular for $N$ large enough 
$$\sup\limits_{0 \leq s \leq t} \lvert ^N\Psi^1_{s,0}(X) - {}^N\Phi^1_{s,0}(X) \rvert_\infty< N^{-\alpha}.$$
Additionally $X \notin \mathcal{C}_t$ from which we can conclude that $|G(\Phi_{t,0}^{N}(X))-\bar{G}(\Phi_{t,0}^{N}(X))|_{\infty}\leq CN^{7\beta-1}\ln(N)$ and in particular
	\begin{align*} 
	\frac{1}{N} \sum\limits_{j=1}^N g_N^{\beta} (\Phi^1_i - \Phi^1_j) &= \bigl(G^N(\Phi_{t,0}(X) \bigr)_i 
	\leq \lVert g_N^{\beta}*\tilde{k}^N_t(q) \rVert_\infty + N^{7\beta-1}\ln(N)\leq C N^{7\beta-1}\ln(N). 
	\end{align*}
	
For the derivative of the momentum we can conclude
	\begin{align*}
	 \frac{\mathrm{d}}{\mathrm{d}t}  \lvert \Psi^2 _t(X) - \Phi^2_{t,0}(X) \rvert_\infty \leq C N^{7\beta-1}\ln(N)\,  \bigl\lvert \Psi^1_{t,0}(X) - \Phi^1_{t,0}(X) \bigr\rvert_\infty + N^{5\beta-1}\ln(N). \end{align*}
	
We observe that for $X\in (A_t\cup B_t\cup C_t)^c$ and $\beta<\frac{1}{7}$
	\begin{align*} &\partial_t^+ \left(\sqrt{\ln(N)} \lvert \Psi^1_{t,0}(X) - \Phi^1_{t,0}(X) \rvert_\infty + \lvert \Psi^2_{t,0}(X) - \Phi^2_{t,0}(X) \rvert_\infty\right)\Bigl\lvert_{\mathcal{A}_t \cap \mathcal{B}_t \cap \mathcal{C}_t}\\
	&\leq\,\sqrt{\ln(N)} \frac{\mathrm{d}}{\mathrm{d}t}  \lvert \Psi^1_{t,0}(X) - \Phi^1_{t,0}(X) \rvert_\infty +   \frac{\mathrm{d}}{\mathrm{d}t}  \lvert \Psi^2_{t,0}(X) - \Phi^2_{t,0}(X) \rvert_\infty\\
	&\leq\,  \sqrt{\ln(N)} \lvert \Psi^2_{t,0}(X) - \Phi^2_{t,0}(X) \rvert_\infty
	 + \,  C \ln(N) \left(N^{7\beta-1} \lvert \Psi^1_{t,0}(X) - \Phi^1_{t,0}(X) \rvert_\infty + N^{5\beta-1} \right)
	\\ \leq \, &\; C\,\sqrt{\ln(N)} \left(\sqrt{\ln(N)} \lvert \Psi^1_{t,0}(X) - \Phi^1_{t,0}(X) \rvert_\infty + \lvert \Psi^2_{t,0}(X) - \Phi^2_{t,0}(X) \rvert_\infty\right) + N^{5\beta-1}.
	\end{align*}
In total the right derivative of $J_t^{N}$ is given by
	$$\partial_t^+ J_t^{N} (X) \leq \max \left\lbrace 0 , \frac{d}{dt}\left(\sigma_{N,t}\left(N^{\alpha} \left(\sqrt{\ln(N)} \lvert \Psi^1_{t,0}(X) - \Phi^1_{t,0}(X) \rvert_\infty + \lvert \Psi^2_{t,0}(X) - \Phi^2_{t,0}(X) \rvert_\infty+N^{5\beta -1}\right)
\right)	\right)  \right\rbrace$$
	with
	\begin{align*}\notag 
		&\frac{d}{dt}\left(\sigma_{N,t}\left(N^{\alpha} \left(\sqrt{\ln(N)} \lvert \Psi^1_{t,0}(X) - \Phi^1_{t,0}(X) \rvert_\infty + \lvert \Psi^2_{t,0}(X) - \Phi^2_{t,0}(X) \rvert_\infty\right)+N^{5\beta -1}
	\right) \right)\\&\leq
			 -\lambda \sqrt{\ln(N)} e^{\lambda \sqrt{\ln(N)}(T-t)}\left(N^{\alpha} \left(\sqrt{\ln(N)} \lvert \Psi^1_{t,0}(X) - \Phi^1_{t,0}(X) \rvert_\infty + \lvert \Psi^2_{t,0}(X) - \Phi^2_{t,0}(X) \rvert_\infty\right)+N^{5\beta -1}\right)\\
		&+ e^{\lambda \sqrt{\ln(N)}(T-t)}N^{\alpha}\left( C \sqrt{\log{N}} \left(\sqrt{\ln(N)} \lvert \Psi^1_{t,0}(X) - \Phi^1_{t,0}(X) \rvert_\infty + \lvert \Psi^2_{t,0}(X) - \Phi^2_{t,0}(X) \rvert_\infty\right) + N^{5\beta -1}\right)
			\\\notag &= \sqrt{\ln(N)}N^{\alpha} e^{\lambda \sqrt{\ln(N)}(T-t)}\Big[\left(C\, -\lambda \right)\left(\sqrt{\ln(N)} \lvert \Psi^1_{t,0}(X) - \Phi^1_{t,0}(X) \rvert_\infty + \lvert \Psi^2_{t,0}(X) - \Phi^2_{t,0}(X) \rvert_\infty\right)\\
			&\hspace{5cm}+ \left((\ln(N))^{-\frac{1}{2}}-\lambda \right) \, N^{5\beta -1}\Big].
	\end{align*}
	By choosing $\lambda=C$ this derivative is negative and thus we get
	$
	  \partial_t^+  \mathbb{E}(J^{N}_t \mid \mathcal{A}_{t}^{c} \cap \mathcal{B}_{t}^{c} \cap \mathcal{C}_{t}^{c}) = 0.
	  $
Finally we can conclude for the right derivative of the expectation value of the auxiliary process
	  $$ \partial_t^+ \mathbb{E}(J^{N}_t )  \leq e^{\lambda \sqrt{\ln(N)}T}C_\gamma N^{-\gamma}. $$
And thus by the linearity of the expectation value the following bound holds
\begin{align*} &\mathbb{E}_0(J^{N}_t) -  \mathbb{E}_0(J^{N}_0)=   \mathbb{E}_0\bigl(J^{N}_t-J^{N}_0 \bigr) \leq T e^{\lambda \sqrt{\ln(N)}T}C_\gamma N^{-\gamma}, 
\end{align*}
uniformly in $t\in[0,T]$.\\
The initial states where chosen such that
 $\left(\sqrt{\ln(N)} \lvert \Psi^1_{0,0}(X) - \Phi^1_{0,0}(X) \rvert_\infty + \lvert \Psi^2_{0,0}(X) - \Phi^2_{0,0}(X) \rvert_\infty\right)=0$ and thus at time $t=0$ the auxiliary process $J^{N}_0(X)\equiv e^{\lambda \sqrt{\ln(N)}T} N^{5\beta+\alpha -1}$ which 
for $N$ sufficient large and $0<\alpha<\beta<\frac{1}{7}$ is bounded by
 \begin{align*}
 e^{\lambda \sqrt{\ln(N)}T} N^{5\beta+\alpha -1}\leq  e^{\lambda \sqrt{\ln(N)}T} N^{5\beta+\alpha -1} \leq \frac{1}{2} 
 \end{align*}
as $e^{\sqrt{\ln(N)}}$ grows slower than than $N^{\epsilon}$ for all $\epsilon>0$. We observe that the random variable $J^{N}_t - J^{N}_0$ is certainly non-negative and it follows that	  
	\begin{align*}\mathbb{P}_0 \Bigl[J^{N}_T(X)-J^{N}_0(X) \geq \frac{1}{2} \Bigr] \leq 2T e^{\lambda \sqrt{\ln(N)}T} C_\gamma N^{-\gamma}. \end{align*}
	In the case $J^{N}_t - J^{N}_0<\frac{1}{2}$ we have $J^{N}_t(X)<1$ and thus we can conclude
\begin{align*}
 \begin{split}&\mathbb{P}_0\Bigl[\sup\limits_{0 \leq s \leq T}\left\{  \left(\sqrt{\ln(N)} \lvert \Psi^1_{t,0}(X) - \Phi^1_{t,0}(X) \rvert_\infty + \lvert \Psi^2_{t,0}(X) - \Phi^2_{t,0}(X) \rvert_\infty\right) \right\} \geq  N^{-\alpha}   \Bigr]\\ &\leq \mathbb{P}_0\Bigl[J^{N}_T(X)-J^{N}_0(X)
 \geq \frac{1}{2} \Bigr] \\&\leq 2T e^{\lambda \sqrt{\ln(N)}T} C_\gamma N^{-\gamma} \leq 2T C_\gamma N^{1-\gamma - 5\beta-\alpha}.
\end{split}
\end{align*}

We can find for any given $\tilde{\gamma} > 0$ by choosing $\gamma := \tilde{\gamma} + 1 - 5\beta-\alpha$ a $C_{\tilde{\gamma}}$ such that 
\begin{align*}
\mathbb{P}_0\Bigl[\sup\limits_{0 \leq s \leq T}\left\{  \left(\sqrt{\ln(N)} \lvert \Psi^1_{t,0}(X) - \Phi^1_{t,0}(X) \rvert_\infty + \lvert \Psi^2_{t,0}(X) - \Phi^2_{t,0}(X) \rvert_\infty\right)\right\}
\geq N^{-\alpha}   \Bigr] \leq T C_{\tilde{\gamma}}N^{-\tilde{\gamma}}.
\end{align*}
This proves Lemma \ref{Mainlemma}.

We are left to show that the we fully approximate the non regularised system and therefore we define  $$\Delta_N(t):=\sup_{x\in\mathbb{R}^6}\sup_{0\leq s\leq T}|\varphi^1_{s,0}(X) - \varphi^\infty_{s,0}(X)|\leq 2N^{-\alpha},$$
for $f^{\infty}=\delta_0(x)$
 which will conclude the proof of Theorem \ref{maintheorem}.
\section{Proof of Theorem \ref{maintheorem}}\label{Proof Section}

Let $t\in [0,T]$ be such that still $\Delta_N(t)\le N^{-\alpha}$, then it holds for $x\in \mathbb{R}^6$ and $N\in \mathbb{N}\setminus \{1\}$ that
\begin{align*}
&\partial_t \sup_{x\in\mathbb{R}^6}|^1\varphi_{t,0}^N(x)-{^1\varphi^{\infty}_{t,0}}(x)|=\sup_{x\in\mathbb{R}^6}|^2\varphi_{t,0}^N(x)-{^2\varphi^{\infty}_{t,0}}(x)|\\
&\leq 
\sup_{q_0,p_0\in\mathbb{R}^3}|\widetilde{k}^N_t\ast f_N^{\beta}(q_t^N(q_0,p_0))-\widetilde{k}^N_t\ast f_N^{\beta}(q_t^\infty(q_0,p_0))| 
\\&+\sup_{q_0,p_0\in\mathbb{R}^3}|\widetilde{k}^N_t\ast f_N^{\beta}(q_t^\infty(q_0,p_0))-\widetilde{k}^N_t\ast f^\infty(q_t^\infty(q_0,p_0))| 
\\&+\sup_{q_0,p_0\in\mathbb{R}^3}|\widetilde{k}^N_t\ast f^\infty(q_t^\infty(q_0,p_0))-\widetilde{k}^\infty_t\ast f^\infty(q_t^\infty(q_0,p_0))|\;. 
\\
&\leq\|\widetilde{k}^N_t\ast\nabla f_N^{\beta}\|\|q_t^N-q_t^\infty\|_\infty 
+\|\widetilde{k}^N_t\ast f_N^{\beta}-\widetilde{k}^N_t\ast f^\infty\|_\infty
+\|\widetilde{k}^N_t\ast f^\infty-\widetilde{k}^\infty_t\ast f^\infty\|_\infty.
\end{align*}
The first addend is bounded by $C\Delta_N(t)$ due to Lemma \ref{Abschätzungen Faltung} and because of the restrictions on $\tilde{k}$
\begin{align*}
\left|\widetilde{k}^N_t(q)-\widetilde{k}^\infty_t(q)\right|
\leq& \int \left|k_0(q^N_t,p^N_t)-k_0(q^\infty_t,p^\infty_t)\right| d^3p_0
\\\leq&\int \left(\sup_{h\in\mathbb{R}^3;|h|=1}\nabla k_0(q^\infty_t,p^\infty_t+h)\right)|x^N_t-x^\infty_t| d^3p_0 
\;.
\end{align*}
for a time $t$ such that $\Delta_N(t)\leq 1$.
The second one is also bounded due to Lemma \ref{Abschätzungen Faltung} by 
\begin{align*}
\|\widetilde{k}^N_t\ast f_N^{\beta}-\widetilde{k}^N_t\ast f^\infty\|_\infty\leq\|\widetilde{k}^N_t\ast (f_N^{\beta}- f^\infty)\|_\infty\leq\|\widetilde{k}^N_t\ast f_N^{\beta}\|_\infty\leq C\Delta_N(t)
\end{align*}
So we get by Gronwalls lemma
$$\sup_{q_0,p_0\in\mathbb{R}^3}|x_s^N(q_0,p_0)-x_s^\infty(q_0,p_0)|\leq C N^{-\alpha}\;,$$
which implies \begin{align*}
\left\|\Phi^N_{s,0}-\Phi^\infty_{s,0}\right\|_\infty<CN^{-\alpha}\;.
\end{align*}
That shows that the initial assumption $\Delta_N(t)\le N^{-\alpha}$ stays true for times $t<T$ provided that $N\in \mathbb{
N}$ is large enough.\\

\section{Molecular chaos}

This result implies molecular chaos in the sense of Corollary \ref{Chaoscor} which is stated below. Therefore we introduce the following notation of distance which we require to estimate the dissimilarity between the two probability measures.

\begin{definition}\label{Def_W}
	Let $\mathcal{P}(\R^n)$ be the set of probability measures on $\R^n$.For $\mu, \nu \in \mathcal{P}(\R^n)$, let $\Pi(\mu, \nu)$ be the set of all probability measures $\R^n \times \R^n$ with marginal $\mu$ and $\nu$. Then, for $p \in [1, +\infty)$, the p'th Wasserstein distance on $\mathcal{P}(\R^n)$ is defined by
	\begin{align*}
	W_p(\mu, \nu) := \inf\limits_{\pi \in \Pi(\mu,\nu)} \, \Bigl( \int\limits_{\R^n\times\R^n} \lvert x - y \rvert^p \, \dd \pi(x,y) \, \Bigr)^{1/p}.
	\end{align*}
For $p=\infty$  the infinite Wasserstein distance is defined by
\begin{align*}
		 W_\infty(\mu, \nu) = \inf \lbrace \pi- \mathrm{esssup}\,  \lvert x - y \rvert \, : \, \pi \in \Pi(\mu,\nu)\rbrace.
		 \end{align*}		
\end{definition}
In particular this notion of distance implies weak convergence in $\mathcal{P}(\R^n)$.
Theorem \ref{maintheorem} implies molecular chaos in the following sense:
\begin{corollary}\label{Chaoscor}
Let $K^N_0 := \otimes^N k_0$ be the $n-$fold solution of the considered effective equation and $K^N_t := \Psi_{t,0}^{N} \# K_0$ the $N$-particle distribution at time $t\in\lbrack0,T\rbrack$ evolving with the microscopic flow \eqref{Def:Mean-fieldflow}. Then the $i$-particle marginal
			\begin{align*} 
			^{(i)}K^N_t (x_1, ..., x_i) := \int K^N_t(X) \, \mathrm{d}^6x_{i+1}...\mathrm{d}^6x_N
			\end{align*}
			converges weakly to $\otimes^i k_t$
			{as} $N \to \infty$ for all $k \in  \N$, where $k_t$ is the unique solution of the Vlasov-Dirac-Benney equation \eqref{Vlasov} with $k^N\lvert_{t=0} = k_0$.  More precisely, under the assumptions of the previous theorem, we get  a constant $C >0$ such that for all $N \geq N_0$
			\begin{align*} 
			W_1(^{(i)}K_t^N, \otimes^i k_t) \leq i \, e^{TC\sqrt{\ln(N)}}N^{-\alpha},\, \forall 0 \leq t \leq T.
			\end{align*}
			 
\end{corollary}
\begin{proof} 
			
	 For a fixed time $0 \leq t \leq T$ and $\mathcal{A} \subset \R^{6N}$ defined in Definition \ref{definition ABC} we have proven in Theorem \ref{maintheorem}, that $\mathbb{P}_0(\mathcal{A}) \leq TC_\gamma N^{-\gamma}$ for sufficiently large $N$. 
	Using the notion of distance above an that all test-functions are Lipschitz with $\lVert h \rVert_{Lip} = 1$,
			\begin{align*}\notag W_1({}^{(i)}K_t^N, &\otimes^i k_t)\\
			=& \sup\limits_{\lVert h \rVert_{Lip}=1}\, \Bigl\lvert\int \bigl(K_t^N (X) -   \otimes^N k_t (X) \bigr) g(x_1,...,x_i) \mathrm{d}^6x_1...\dd^6{x_k}...\mathrm{d}^6x_N\Bigr\rvert \\
			= &\sup\limits_{\lVert h \rVert_{Lip}=1}\, \Bigl\lvert \int \bigl(\Psi_{t,0}\#K^N_0 (X) -  \Phi_{t,0}\#K^N_0(X)\bigr) g(x_1,...,x_i) \,\mathrm{d}^{6N}X\Bigr\rvert\\
			= &\sup\limits_{\lVert h \rVert_{Lip}=1}\, \Bigl\lvert \int K^N_0(X) \bigl(h(\pi_i \Psi_{t,0}(X)) - h(\pi_i\Phi_{t,0}(X))\bigr) \,\mathrm{d}^{6N}X\Bigr\rvert\\
			=&\sup\limits_{\lVert h \rVert_{Lip}=1}\, \Bigl\lvert \int\limits_{\mathcal{A}} K^N_0(X) \bigl(h(\pi_i \Psi_{t,0}(X)) - h(\pi_i \Phi_{t,0}(X))\bigr) \,\mathrm{d}^{6N}X\Bigr\rvert \\
			+& \sup\limits_{\lVert h \rVert_{Lip}=1}\, \Bigl\lvert \int\limits_{\mathcal{A}^c} K^N_0(X) \bigl(h(\pi_i \Psi_{t,0}(X)) - h(\pi_i \Phi_{t,0}(X))\bigr) \,\mathrm{d}^{6N}X\Bigr\rvert
			\end{align*}
for the projection $\pi_i: \R^{N} \to \R^i, (x_1,...,x_N) \mapsto (x_1,...,x_i)$.		
The first addend is bounded by
$$\sup\limits_{\lVert h \rVert_{Lip}=1}\, \Bigl\lvert \int\limits_{\mathcal{A}} K^N_0(X) (h\bigl(\pi_i \Psi_{t,0}(X)) - h(\pi_i \Phi_{t,0}(X))\bigr) \,\mathrm{d}^{6N}X\Bigr\rvert\leq  \mathbb{P}(\mathcal{A}^c) \lVert K^N_0\rVert_\infty \lvert \Psi_{t,0}^{N}(X) - \Phi_{t,0}^{N}(X) \rvert,$$
 with $\lVert K^N_0\rVert_\infty = (\lVert k_0 \rVert_\infty)^N$. By the initialization of the initial conditions we trivially have that  $\lvert  {}^N\Psi_{0,0}(X) - {}^N\Phi_{0,0}(X) \rvert_\infty = 0$ and by Newtons law
			\begin{align*}
			\lvert \Psi^{2,N}_{t,0}(X) - \Phi^{2,N}_{t,0}(X) \rvert_\infty &\leq \int\limits_{0}^t \lvert F^N(\Psi^1_{s,0}(X)) - \overline{F}(\Phi^1_{s,0}(X)) \rvert_\infty \dd s,\\
				\lvert \Psi^{1,N}_{t,0}(X) - \Phi^{1,N}_{t,0}(X) \rvert_\infty &\leq \int\limits_{0}^t  \lvert  \Psi^{2,N}_{s,0}(X) - \Psi^{2,N}_{s,0}(X) \rvert_\infty \dd s.
			\end{align*}
	 The mean-field force $\overline{F}$ is of order 1 and the microscopic force $F^N$ is bounded by $N^{4\beta}$. Hence, there exists a constant $C>0$ such that $\lvert \Psi^{2,N}_{t,0}(X) - \Phi^{2,N}_{t,0}(C) \rvert_\infty \leq T CN^{4\beta}$ and consequently  to Newtons law $\lvert \Psi^{1,N}_{t,0}(X) - \Phi^{1,N}_{t,0}(X) \rvert_\infty \leq T^2CN^{4\beta}$ for all $t \leq T$. Choosing $\gamma := 5\beta$ in Theorem \ref{maintheorem} we thus get $C$ such that
			\begin{align*}\mathbb{P}(\mathcal{A}^c) \lVert K^N_0\rVert_\infty \lvert \Psi_{t,0}^{N}(X) - \Phi_{t,0}^{N}(X) \rvert \leq C \max \lbrace T^2, T^3 \rbrace N^{-\beta},
			\end{align*}
			for all times $ 0 \leq t \leq T$.
			 On the other hand, for $X \in \mathcal{A}^c$, we have for any $h$ with $\lVert h \rVert_{Lip} = 1$, 
				\begin{equation*} 
			   \lvert h(\pi_i \Psi_{t,0}^{N}(X)) - h(\pi_i \Phi_{t,0}^{N}(X)) \rvert \leq  \lvert  \Psi_{t,0}^{N}(X) - \Phi_{t,0}^{N}(X) \rvert_\infty \leq N^{-\alpha}
				\end{equation*}
			 for all $t \leq T$ and thus  $$\sup\limits_{\lVert g \rVert_{Lip}=1}\, \Bigl\lvert \int\limits_{\mathcal{A}^c} K^N_0(X) \bigl(h(\pi_i \Psi_{t,0}(X)) - h(\pi_i \Phi_{t,0}(X))\bigr) \,\mathrm{d}^{6N}X\Bigr\rvert\leq N^{-\alpha}.$$ It follows that there exists  a constant $C$ such that
				\begin{align*}
				W_1(^{(i)}K_t^N, \otimes^i k^N_t) \leq C(1+T^3)N^{-\alpha},
				\end{align*}
for all times $ 0 \leq t \leq T$ and due to the property that $k^N$ approximate $k$ (see \cite[Prob. 9.1]{Dustin} the statement follows. 
		\end{proof}
		Molecular chaos in the sense of Corollary \ref{Chaoscor} implies convergence in law of the empirical distribution to the solution of the Vlasov Dirac Benney equation  $k_t$ (see e.g.  \cite{Kac}, \cite{Grunbaum}, \cite[Prop.2.2]{Sznitman}).
Finally one can derive the macroscopic mean-field equation \eqref{Vlasov} from the microscopic random particle system \ref{Def:Newtonflow}. We define the empirical measure associated to the microscopic $N$-particle systems and respectively to the macroscopic by
\begin{align*}
\mu_\Phi(t):=\frac{1}{N}\sum_{i=1}^{N}\delta(q-q_i^t)\delta(p-p_i^t),\quad \mu_\Psi(t):=\frac{1}{N}\sum_{i=1}^{N}\delta(q-\overline q_i^t)\delta(p-\overline p_i^t)
\end{align*}
and will see that the empirical measure $\mu_\Phi(t)$ converges to the solution of the Vlasov-Dirac-Benney equation in $W_p$ distance with high probability.

	\begin{theorem}[Particle approximation of the Vlasov-Dirac-Benney system]\label{Vlasov approx}
	Let $k_0$ be a probability measure satisfying the assumptions of Theorem \ref{maintheorem} ${}^N\Psi_{t,s}$ be the $N$-particle flow solving \eqref{Def:Newtonflow}.  Then, the empirical density $\mu_{\Phi_0(t)}$ converges to the solution of the Vlasov-Dirac-Benney equation in the following sense:\\
	
	\noindent For any $T>0$ there exists a constant $C$ depending on $k_0$ and $T$ such that for all $N \geq N_0$ and some $\eta,\iota>0$
	\begin{align*}
	\mathbb{P}\Bigl[ \max_{t\in\lbrack 0,T\rbrack } \in [0,T] : W_p(\mu_\Phi, k_t) >  N^{-\eta}  \Bigr] 
	\leq  C e^{-CN^{1-\iota}},
\end{align*}
	\noindent where $k$ is the unique solution of the Vlasov-Dirac-Benney system on $[0,T]$. 
\end{theorem}
\begin{proof}
In order to prove this let us split $W_p(\mu_\Phi(t),k_t)$ into three parts
\begin{align*}
W_p(\mu_\Phi(t),k_t)&\leq W_p(k_t,k_t^N)+W_p(k_t^N,\mu_\Psi(t))+W_p(\mu_\Psi(t),\mu_\Phi(t)).
\end{align*}
The convergence of the first addend is a deterministic result stated in \cite[Prob. 9.1]{Dustin}, the second addend is bounded in probability due to\cite[Cor. 9.4]{Dustin}
 and the last term is bounded in probability due to Theorem \ref{maintheorem}.
\end{proof}

\bibliographystyle{amsplain}

\newpage		

\begin{thebibliography}{10}
\bibitem{Bardos Besse}C. Bardos and N. Besse. Hamiltonian structure, fluid representation and stability for the
Vlasov-Dirac-Benney equation. In Hamiltonian partial differential equations and applications,
volume 75 of Fields Inst. Commun., pages 1–30. Fields Inst. Res. Math. Sci., Toronto, ON,
2015.
\bibitem{Bardos-Besse}
C.~Bardos and N.~Besse.
\newblock {The Cauchy problem for the Vlasov-Dirac-Benney equation and related
  issued in fluid mechanics and semi-classical limits}.
\newblock {\em Kinet. Relat. Models}, 6(4):893--917, 2013.

\bibitem{Benachour}
S.~Benachour,
\newblock Analyticité des solutions des équations de Vlasov-Poisson. 
\newblock In Annali della Scuola Normale Superiore di Pisa - Classe di Scienze, Serie 4, Volume 16 (1989) no. 1, pp. 83-104.
\bibitem{Nouri} P.~-E.~ JABIN AND A.~ NOURI, Analytic solutions to a strongly nonlinear Vlasov equation, C. R. Math. Acad.
Sci. Paris, 349 (2011), pp. 541–546.

\bibitem{Besse}
N.~Besse,F.~ Berthelin, Y.~ Brenier,P.~ Bertrand, The multi water-bag model for collisionless
kinetic equations. Kinetic and Related Models 2, nb. 1 (2009) 39-80.
\bibitem{Peter}
N.~Boers and P.~Pickl.
\newblock On mean-field limits for dynamical systems.
\newblock {\em Journal of Statistical Physics}, 164(1):1--16, 2016.

\bibitem{BraunHepp}
W.~Braun and K.~Hepp.
\newblock The {V}lasov dynamics and its fluctuations in the 1/{N} limit of
  interacting classical particles.
\newblock {\em Communications in Mathematical Physics}, 56(2):101--113, 1977.






\bibitem{Dobrushin}
R.~L.~ Dobrushin.
\newblock Vlasov equations.
\newblock {\em Functional Analysis and Its Applications}, 13(2):115--123, 1979.

\bibitem{Goodman}
J.~Goodman.
\newblock Convergence of the random vortex method.
\newblock {\em Communications on Pure and Applied Mathematics}, 40(2):189--220,
  1987.
\bibitem{grass}
P.~Gra{\ss}.
\newblock {\em Microscopic derivation of Vlasov equations with singular
  potentials}.
\newblock PhD thesis, lmu, 2019.
  \bibitem{Grunbaum}
F.~A.~ Gr{\"u}nbaum.
\newblock Propagation of chaos for the {B}oltzmann equation.
\newblock {\em Archive for Rational Mechanics and Analysis}, 42:323--345, 1971.

\bibitem{HanKwanRousset}
D.~ Han-Kwan,F.~ Rousset.
\newblock Quasineutral limit for Vlasov-Poisson with Penrose stable data.
\newblock Societe Mathematique de France, Paris, 4e serie, t. 49, 2016, p. 1445-1495

\bibitem{Hauray2013}
M.~Hauray and P.~ E.~ Jabin.
\newblock Particles approximations of {V}lasov equations with singular forces : Propagation of chaos.
\newblock To appear in \textit{Ann. Sci. Ec. Norm. Super., s{\'e}rie} 48:891--940, 2015.
\bibitem{Hauray2007}
M.~ Hauray and P.~ E.~ Jabin.
\newblock N-particles approximation of the Vlasov equations with singular potential.
\newblock {\em Arch. Ration. Mech. Anal.}, 183(3):489-â, 2007.
\bibitem{Kac}
M.~Kac.
\newblock Foundations of kinetic theory.
\newblock In {\em Proceedings of the Third Berkeley Symposium on Mathematical
  Statistics and Probability, 1954-1955}, Vol. III, pages 171--197.
  University of California Press, 1956.
  \bibitem{Kiessling}
M.~K.-H. Kiessling.
\newblock The microscopic foundations of {V}lasov theory for jellium-like
  {N}ewtonian {N}-body systems.
\newblock {\em Journal of Statistical Physics}, 155(6):1299--1328, 2014.
\bibitem{Dustin}
D.~ Lazarovici. 
\newblock The Vlasov-Poisson dynamics as the mean-field limit of extended charges. 
\newblock{ \em Communications in Mathematical Physics}, 347(1):271--289, 2016.


\bibitem{NeunzertWick}
H.~Neunzert and J.~Wick.
\newblock Die {A}pproximation der {L}\"{o}sung von {I}ntegro-{D}ifferentialgleichungen durch endliche {P}unktmengen.
\newblock In {\em Numerische Behandlung nichtlinearer {I}ntegrodifferential - und Differentialgleichungen}, volume 395 of {\em Lecture Notes in Mathematics}, pages 275--290. Springer, Berlin, Heidelberg, 1974.



\bibitem{Oelschläger Hydro}
K.~Oelschläger 
\newblock {\em On the connection between Hamiltonian many-particle systems and the hydrodynamical equations}. 
\newblock Arch. Rational Mech. Anal. 115, 297–310 (1991).
\bibitem{Oelschläger Brown}
K.~A.~Oelschläger, law of large numbers for moderately interacting diffusion processes. Z. Wahrscheinlichkeitstheorie verw Gebiete 69, 279–322 (1985).
\bibitem{Penrose}
O.~{Penrose}.
\newblock {Electrostatic instability of a uniform non-Maxwellian plasma}.
\newblock {\em Phys. Fluids}, 3:258--265, 1960.
 
\bibitem{SpohnBook}
H.~Spohn.
\newblock  Large scale dynamics of interacting particles. Springer, Berlin, Heidelberg, 1991.
  \bibitem{Sznitman}
A.~S~. Sznitman.
\newblock Topics in propagation of chaos.
\newblock In {\em {\'E}cole d'{\'E}t{\'e} de Probabilit{\'e}s de Saint-Flour
  XIX -- 1989}, volume 1464 of {\em Lecture Notes in Mathematics}, pages
  165--251. Springer, Berlin, 1991.





\bibitem{VLASOV}
A.~ A.~ Vlasov (1938) "On Vibration Properties of Electron Gas". J. Exp. Theor. Phys. (in Russian). 8 (3): 291.

  \bibitem{Zhang}
Z.~Chen, Z. and X. Zhang. 
\newblock Global existence to the Vlasov-Poisson system and propagation of moments without assumption of finite kinetic energy.


\end{thebibliography}

\end{document}